\documentclass[10pt,twocolumn,letterpaper,table]{article}
\usepackage{usenix-2020-09}

\newif\ifusenixsty
\usenixstytrue
\newif\ifcameraready
\camerareadyfalse
\newif\ifremovecomments
\removecommentstrue
\newif\ifmarkchanges
\markchangestrue

\usepackage{lmodern}
\usepackage{amssymb,amsmath,amsfonts}
\usepackage[T1]{fontenc}
\usepackage[utf8]{inputenc}

\usepackage{upquote}

\usepackage{microtype}
\UseMicrotypeSet[protrusion]{basicmath} 
\usepackage{hyperref}
\usepackage{graphicx,grffile}

\usepackage{times}
\usepackage{tikz}
\usetikzlibrary{shapes.misc, shapes.geometric, shapes.multipart, shapes.arrows, positioning, quotes, arrows.meta, decorations.pathmorphing, matrix, graphs, fit, calc, automata, positioning}    
\tikzset{
	state/.style={
		rectangle,
		rounded corners,
		draw=black, thick,
		minimum height=2em,
		inner sep=2pt,
		text centered,
	},
}
\usepackage{endnotes,epsfig}
\usepackage{multirow}
\usepackage{xspace}
\usepackage{tabularx}
\usepackage{ragged2e}
\usepackage{booktabs}
\usepackage{paralist}
\usepackage[american]{babel}
\usepackage[shortlabels]{enumitem}
\usepackage{courier,color,wrapfig}
\usepackage{xspace}
\usepackage{balance}
\usepackage{enumitem}
\usepackage{epstopdf}
\usepackage{multirow}
\usepackage{booktabs}
\usepackage{url}
\def\UrlBreaks{\do\/\do-} 
\usepackage{amssymb}
\usepackage{pifont}
\usepackage{subfigure}
\usepackage{comment}
\usepackage{mathtools}
\usepackage[normalem]{ulem}
\usepackage{xkeyval}
\usepackage{cleveref}
\usepackage{xcolor,colortbl}

\ifremovecomments
\usepackage[disable, textsize=small]{todonotes}
\else
\usepackage[textsize=small]{todonotes}
\fi

\usepackage{ifthen}
\usepackage{etoolbox}

\usepackage[margin=5pt,font=small,labelfont={bf,rm}]{caption}
\DeclareCaptionFormat{myformat}{#1#2#3\hrulefill}
\usepackage{sepfootnotes}
\usepackage[compact]{titlesec}
\usepackage{bbding}
\usepackage{enumitem}

\usepackage{caption}
\DeclareCaptionFormat{myformat}{#1#2#3\hrulefill}
\usepackage{array}
\newcolumntype{L}[1]{>{\raggedright\let\newline\\\arraybackslash\hspace{0pt}}m{#1}}
\newcolumntype{C}[1]{>{\centering\let\newline\\\arraybackslash\hspace{0pt}}m{#1}}
\newcolumntype{R}[1]{>{\raggedleft\let\newline\\\arraybackslash\hspace{0pt}}m{#1}}
\usepackage{amsmath,amsthm}
\makeatletter
\def\thm@space@setup{%
	\thm@preskip=6pt plus 1pt minus 3pt
	\thm@postskip=\thm@preskip
}
\makeatother
\usepackage{bbm}
\usepackage{algorithm}
\usepackage{algpseudocode}

 \usepackage{bigdelim}

\makeatletter
\def\maxwidth{\ifdim\Gin@nat@width>\linewidth\linewidth\else\Gin@nat@width\fi}
\def\maxheight{\ifdim\Gin@nat@height>\textheight\textheight\else\Gin@nat@height\fi}
\makeatother
\setkeys{Gin}{width=\maxwidth,height=\maxheight,keepaspectratio}
\makeatletter
\g@addto@macro{\UrlBreaks}{\UrlOrds}
\makeatother

\newcommand\paraspace{\vspace*{1ex}}

\providecommand\parab[1]{\paraspace\noindent\textbf{#1}}
\providecommand\parae[1]{\paraspace\textbf{\textit{#1}}}

\apptocmd\normalsize{%
	\abovedisplayskip=5pt
	\abovedisplayshortskip=5pt
	\belowdisplayskip=5pt
	\belowdisplayshortskip=5pt
}{}{}

\newcommand{\dempin}{{\small\textsf{DP}}\xspace}
\newcommand{\pop}{{\small\textsf{POP}}\xspace}
\newcommand{\opt}{{\small$\textsf{H}'$}\xspace}
\newcommand{\eqopt}{{\small\textsf{H}'}\xspace}
\newcommand{\heur}{{\small\textsf{H}}\xspace}
\newcommand{\metaopt}{\textsf{MetaOpt}\xspace}
\newcommand{\vbp}{\textsf{VBP}\xspace}
\newcommand{\ffd}{\textsf{FFD}\xspace}
\newcommand{\ffdsum}{\textsf{FFDSum}\xspace}
\newcommand{\sppifo}{\textsf{SP-PIFO}\xspace}
\newcommand{\aifo}{\textsf{AIFO}\xspace}
\newcommand{\pifo}{\textsf{PIFO}\xspace}
\newcommand{\qpd}{\textsf{QPD}\xspace}
\definecolor{ForestGreen}{RGB}{34,139,34}

\definecolor{mygreen}{RGB}{208, 244, 222}
\definecolor{myp}{RGB}{253, 226, 228}
\definecolor{myb}{RGB}{198, 222, 241}

\definecolor{mymagenta}{RGB}{255, 198, 255}
\definecolor{myo}{RGB}{255, 214, 165}

\algnewcommand{\IIf}[1]{\State\algorithmicif\ #1\ \algorithmicthen}
\algnewcommand{\EndIIf}{\unskip\ \algorithmicend\ \algorithmicif}
\algblock{Input}{EndInput}
\algnotext{EndInput}
\algblock{Output}{EndOutput}
\algnotext{EndOutput}

\algblock{OuterVar}{EndOuterVar}
\algnotext{EndOuterVar}

\algdef{SE}[IF]{CIF}{ENDIF}%
[2][default]{\colorbox{myp}{\algorithmicif\ #2\ \algorithmicthen}}%
{\algorithmicend\ \algorithmicif}%
\algdef{Ce}[ELSE]{IF}{CELSE}{ENDIF}%
[1][default]{\colorbox{myp}{\algorithmicelse}}

\algdef{SE}[IF]{CIFG}{ENDIF}%
[2][default]{\colorbox{mygreen}{\algorithmicif\ #2\ \algorithmicthen}}%
{\algorithmicend\ \algorithmicif}%
\algdef{Ce}[ELSE]{IF}{CELSEG}{ENDIF}%
[1][default]{\colorbox{mygreen}{\algorithmicelse}}%

\newcommand{\ie}{\emph{i.e.,}\xspace}
\newcommand{\eg}{\emph{e.g.,}\xspace}

\newcommand{\secref}[1]{\S\ref{#1}}
\newcommand{\figref}[1]{Fig.~\ref{#1}}

\newcommand{\tabref}[1]{Table~\ref{#1}}

\newcommand{\eqnref}[1]{Equation~\ref{#1}}

\newcommand{\theoremref}[1]{Theorem~\ref{#1}}

{\begin{itemize}
		\setlength{\leftmargin}{1em}%
		\setlength{\itemsep}{0in}%
		\setlength{\topsep}{-.1in}%
		\setlength{\partopsep}{-.1in}%
		\setlength{\parsep}{-.1in}%
		\setlength{\parskip}{-0in}}%
	{\end{itemize}}

{\begin{enumerate}
		\setlength{\leftmargin}{1em}%
		\setlength{\itemsep}{0in}%
		\setlength{\topsep}{-.1in}%
		\setlength{\partopsep}{-.1in}%
		\setlength{\parsep}{-.1in}%
		\setlength{\parskip}{-0in}}%
	{\end{enumerate}}

\newcommand{\cut}[1]{}

\newcommand{\squishlist}{
	\begin{list}{$\bullet$}
		{ \setlength{\itemsep}{0pt}
			\setlength{\parsep}{3pt}
			\setlength{\topsep}{3pt}
			\setlength{\partopsep}{0pt}
			\setlength{\leftmargin}{1.5em}
			\setlength{\labelwidth}{1em}
			\setlength{\labelsep}{0.5em} } }
	
	\newcommand{\pgftextcircled}[1]{
		\setbox0=\hbox{#1}%
		\dimen0\wd0%
		\divide\dimen0 by 2%
		\begin{tikzpicture}[baseline=(a.base)]%
			\useasboundingbox (-\the\dimen0,0pt) rectangle (\the\dimen0,1pt);
			\node[circle,draw,outer sep=0pt,inner sep=0.1ex] (a) {#1};
		\end{tikzpicture}
	}

	\newcommand{\squishend}{
\end{list}  }

\iffalse
\newcommand{\ramesh}[1]{}
\newcommand{\behnaz}[1]{}
\newcommand{\pooria}[1]{}
\newcommand{\srikanth}[1]{}
\newcommand{\ryan}[1]{}
\newcommand{\santi}[1]{}
\else
\newcommand{\ramesh}[1]{\todo[author=RG,color=blue!30,inline]{#1}}
\newcommand{\behnaz}[1]{\todo[author=BA,color=red!30,inline]{#1}}
\newcommand{\pooria}[1]{\todo[author=PN,color=green!30,inline]{#1}}
\newcommand{\srikanth}[1]{\todo[author=SK,color=cyan!30,inline]{#1}}
\newcommand{\ryan}[1]{\todo[author=RB,color=orange!30,inline]{#1}}
\newcommand{\santi}[1]{\todo[author=SS,color=yellow!30,inline]{#1}}
\fi

\ifmarkchanges
	
	\newcommand{\crdel}[1]{\sout{\textcolor{red}{#1}}}
\else
	
	\newcommand{\crdel}[1]{}
\fi

\newcommand{\indicator}[1]{\mathbbm{1}\left[{#1}\right]}

 \newtheorem{theorem}{Theorem}

\newcommand{\xref}[1]{\S\ref{#1}}

\begin{document}
	
\let\textcircled=\pgftextcircled
\title{Finding Adversarial Inputs for Heuristics using Multi-level Optimization}
	
	
	\author{
	{\rm Pooria Namyar$^{\dag\ddag}$, Behnaz Arzani$^{\dag}$, Ryan Beckett$^{\dag}$, Santiago Segarra$^{\dag\star}$,} \\ 
	{\rm Himanshu Raj$^{\dag}$, Umesh Krishnaswamy$^{\dag}$, Ramesh Govindan$^{\ddag}$, Srikanth Kandula$^{\dag}$} \vspace{1mm}\\
	$^{\dag}$Microsoft, $^{\ddag}$University of Southern California, $^\star$Rice University}

\date{}
	
	
	
\maketitle
	
\noindent {\bf Abstract--}
Production systems use heuristics because they are faster or scale better than their optimal counterparts.
Yet, practitioners are often unaware of the performance gap between a heuristic and the optimum or between two heuristics in realistic scenarios.
We present \metaopt, a system that helps analyze heuristics.
Users specify the heuristic and the optimal (or another heuristic) as input, and {\metaopt} automatically encodes these efficiently for a solver to find performance gaps and their corresponding adversarial inputs.
Its suite of built-in optimizations helps it scale its analysis to practical problem sizes.
To show it is versatile, we used \metaopt to analyze heuristics from three domains (traffic engineering, vector bin packing, and packet scheduling).
We found a production traffic engineering heuristic can require 30\% more capacity than the optimal to satisfy realistic demands.
Based on the patterns in the adversarial inputs \metaopt produced, we modified the heuristic to reduce its performance gap by 12.5$\times$.
We examined adversarial inputs to a vector bin packing heuristic and proved a new lower bound on its performance.

\vspace{-2mm}
\section{Introduction}
\label{s:intro}

Many solutions to network and systems problems are heuristic approximations to potentially intractable optimal algorithms~\cite{ncflow,swan,b4,pop,arrow,tetris,synergy,teavar,soroush}.
These heuristics are often faster or scale better than their optimal counterparts.
However, operators often do not fully understand how these heuristics will behave with new and untested inputs or how far from the optimal their results may drift in realistic use.



For example, Microsoft uses a heuristic, demand pinning (\dempin), on its wide area network~\cite{blastshield,metaopt}.
It routes small demands (\ie demands $\leq$ a threshold) through their shortest path and applies a more computationally complex optimization~\cite{swan}, which jointly considers all the remaining demands.
We can construct an example (see \figref{f:pinning_issues}) where \dempin allocates 40\% less demand than the optimal routing.
A $40\%$ gap is a lower bound~--~the worst case gap can be higher.
With this gap, Microsoft may have to either over-provision their network by 40\%; delay 40\%; or drop 40\% of their customers' demand!

We do not understand the potential impact of such heuristics at scale: Does their gap depend on topology size? For what inputs do they perform poorly? Are there realistic inputs for which they perform poorly?
We ask if we can develop techniques to \emph{analyze heuristics} and answer such questions.


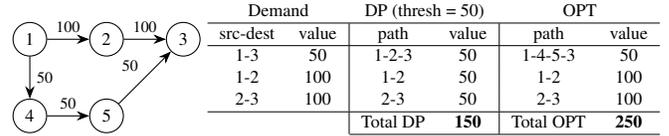
\begin{figure}  
	\centering  
	\begin{tikzpicture}[  
	scale=0.75, transform shape,  
	node distance=0.75cm,  
	font=\small,  
	every edge/.style={draw, -Stealth},  
	every edge quotes/.style={auto, font=\footnotesize},  
	state/.style={circle, draw}  
	]  
	\begin{scope}  
	\node[state] (1) {1};  
	\node[state, right=of 1] (2) {2};  
	\node[state, right=of 2] (3) {3};  
	\node[state, below=of 1] (4) {4};  
	\node[state, right=of 4] (5) {5};  
	
	\path  
	(1) edge[above, "100"] (2)  
	(2) edge[above, "100"] (3)  
	(4) edge[below, "50"] (5)  
	(1) edge[below, "50"] (4)  
	(5) edge[below, "50"] (3)  
	;  
	\end{scope}  
	
	\node[anchor=west] at ([yshift=-0.5cm]3.east){  
	\begin{tabular}{cc | cc | cc}  
	\multicolumn{2}{c}{Demand} & \multicolumn{2}{c}{DP (thresh = 50)} & \multicolumn{2}{c}{OPT} \\  
	\hline  
	src-dest & value & path & value & path & value \\  
	\hline  
	1-3 & 50 & 1-2-3 & 50 & 1-4-5-3 & 50 \\  
	1-2 & 100 & 1-2 & 50 & 1-2 & 100 \\  
	2-3 & 100 & 2-3 & 50 & 2-3 & 100 \\  
	\hline
	& & Total DP & \textbf{150} & Total OPT & \textbf{250}  \\
	\cline{3-6}
	\end{tabular}  } ; 
	\end{tikzpicture}  
\caption{Suboptimal performance of \dempin. (left) Topology with unidirectional links. (right) A set of demands and their flow allocations using the \dempin heuristic and the optimal ({\opt}) solution. \dempin first sends the demands at or below the threshold ($50$) over their shortest paths and then optimally routes the remaining demands. \label{f:pinning_issues}}
\end{figure}   



As a first step towards this goal, we have developed {\metaopt}, a system which allows users to \textit{automatically discover the performance gap} between a heuristic \heur and any other function \opt (which can be the optimal) for much larger problem sizes than in \figref{f:pinning_issues}.
\metaopt also outputs \textit{adversarial inputs} to these functions that cause large performance gaps.
Users can use \metaopt to (\secref{s:background}): (a) understand the performance gap of \heur relative to the optimal or to another heuristic; (b) examine adversarial inputs to provide performance bounds on \heur or to modify \heur to improve its performance gap.

In many problem domains, such as traffic engineering~\cite{swan,b4,blastshield,pop}, vector bin packing~\cite{VBP-Heuristics,tetris}, and packet scheduling~\cite{pifo,aifo,sp-pifo}, we can specify both \opt and \heur either as an optimization problem (with an objective and several constraints) or as a feasibility problem (with a set of constraints).
\metaopt permits users to specify \opt and \heur in one of these forms.
In the language of optimization theory, we can model the problem of finding large performance gaps between \opt and \heur as:

\begin{equation}
\arg\hspace{-12mm}\max_{{\text{s.t. input} \,\, \mathcal{I} \in {\sf ConstrainedSet}}}   {\eqopt}(\mathcal{I}) - {\heur}(\mathcal{I}),
\label{eqn:adversarial_outer}
\end{equation}
where \opt() and $\heur()$ take $\mathcal{I}$ as input and solve the optimal and heuristic algorithms. ${\sf ConstrainedSet}$ specifies a set of constraints that limit the set of values $\mathcal{I}$ can take.
In theory, we could throw this model at modern solvers~\cite{gurobi,zen} and obtain performance gaps.

However, in practice, modern solvers do not support these optimizations. 
This model is an instance of a bi-level optimization~\cite{GentleBilevel} (with connections to Stackelberg games~\cite{games}, see~\S\ref{sec::stackelberg})
, and practitioners have to re-write them into single-level optimizations before using a solver~\cite{GentleBilevel}.
%
%
Rewriting a bi-level optimization into a single-level one by hand is tedious and can lead to poor performance if done incorrectly.

{\metaopt} abstracts away this complexity~---~users only specify \opt and \heur (\secref{s:design}).
It also provides helper functions to make it even easier to specify \opt and \heur. These functions are especially useful for constructs (\eg conditionals, randomization) that are harder to express as convex optimization constraints.
Under the hood, \metaopt performs rewrites automatically, and supports multiple solvers (Gurobi~\cite{gurobi} and Zen~\cite{zen}).
We add three techniques that enable it to scale to large problem sizes (\eg large topologies and demands for traffic engineering).
First, since many rewrites can introduce non-linearities, we carefully selects which part of the input to rewrite.
Second, we introduce a new rewriting technique (Quantized Primal-Dual) that allows \metaopt to trades-off between scale and optimality.
Third, we design a new partitioning technique for graph-based problems to further improve its ability to scale.

We have applied \metaopt to several heuristics in traffic engineering~\cite{swan,b4,blastshield,pop}, vector bin packing~\cite{VBP-Heuristics,tetris}, and packet scheduling~\cite{pifo,aifo,sp-pifo}.
To demonstrate its versatility,we have used \metaopt to (a) study performance gaps of these heuristics, (b) analyze adversarial inputs to prove properties, and (c) devise and evaluate new heuristics.
\tabref{t:metaopt-results} summarizes our results:
\begin{itemize}[nosep,leftmargin=*]
\item \dempin allocates 33\% less demand compared to the optimal in large topologies. We analyzed adversarial inputs \metaopt found and designed modified-\dempin which reduced this gap by an order of magnitude.
\item For a two-dimensional vector bin packing heuristic \ffdsum, we show, for the first time, that it can require at least 2$\times$ as many bins as the optimal {\em across all problem sizes}.
\item A recently proposed programmable packet scheduler, \sppifo, an approximation of PIFO~\cite{pifo}, can delay the highest priority packet by at least 3$\times$ relative to PIFO.
\end{itemize}

\section{Heuristic Analysis at a Glance}
\label{s:background}

Network and systems designers use heuristics to solve problems either when they are NP-hard or when the optimum is too expensive to compute at relevant problem scales.
We use examples of heuristics to motivate what kinds of analyses designers might wish to perform on heuristics, then describe how \metaopt can aid these analyses.


\vspace{-0.7mm}
\subsection{Heuristics and their Importance}
\label{s:heur-their-import}

We describe heuristics from traffic engineering, cluster resource allocation, and switch packet scheduling.

\parab{Traffic engineering (TE).} There are different techniques to scale TE solutions to large networks/demands:

\parae{Demand Pinning (\dempin)}~\cite{blastshield} (~\S\ref{s:intro}), which Microsoft uses in production, pre-allocates flows along the shortest path for any node pair with demand smaller than a threshold $T_d$ and uses the SWAN~\cite{swan} optimizer on the rest.
%
%
When many demands are small, this can result in substantial speedup.


\parae{Partitioned Optimization Problems (\pop)}~\cite{pop} divides node pairs (and their demands) uniformly at random into partitions.
It then assigns each partition an equal share of the edge capacities and solves the original problem (\eg the SWAN LP optimization~\cite{swan}) once per partition.
\pop is faster than SWAN because it can solve each LP sub-problem much faster than the original~\cite{boyd_co} and can do so in parallel.
%


\parab{Vector bin packing (\vbp).}
Production deployments use \vbp to efficiently allocate resources in clusters~\cite{protean,tetris,synergy,yarn_vbp,jobpacking}.
One version of \vbp takes a set of balls and bins with specific sizes and multiple dimensions~(\eg memory, CPU, GPU~\cite{drf,borg,pond}) and tries to pack the balls into the fewest number of bins.
The optimal algorithm for this version is APX-hard~\cite{vbp-nphard}.
Many practitioners use a first fit decreasing~(\ffd) heuristic.
This is greedy and iterative: in each step, it picks the unassigned ball with the {\em largest weight} and places it on the {\em first} bin that {\em fits}  (that has enough capacity).
Prior works propose different ways to weight the balls~\cite{VBP-Heuristics,ffdprod-sandpiper,ffdsum-example,ffddiv-ibm,tetris,protean,synergy}.
One variant, \ffdsum, uses the sum of the dimensions across the items as the weights.

\parab{Packet scheduling}~\cite{sp-pifo,aifo}. Push-In-First-Out (PIFO~\cite{pifo}) queuing is a scheduling primitive that enables various packet scheduling algorithms for programmable switches.
\sppifo~\cite{sp-pifo} uses $n$ priority FIFO queues to approximate PIFO and presents a heuristic which we can implement at line rate.
A packet can only enter a queue if the priority of the queue, which is {\em usually} equal to the priority of the last item in it, is higher than the item's priority.
The algorithm scans queues from lowest to highest priority and places the item in the first queue that accepts it.
\sppifo updates all queue priorities if no queue can admit the packet.

\parab{Performance Analyses.} Often, heuristic designers wish to answer questions such as:
\begin{itemize}[nosep,leftmargin=*]
\item How far is my heuristic from the optimum?
\item What inputs cause my heuristic's performance to degrade at practical problem instances?
\item How can I re-design my heuristic to improve its performance?
\item How can I compare the performance of two heuristics?
\item Can I prove lower bounds on a heuristic's performance relative to another standard algorithm?
\end{itemize}



\subsection{\metaopt, a Heuristic Analyzer}
\label{s:meta-heur-analyz}

\metaopt is a heuristic analyzer that can help designers answer these questions.
\metaopt finds the \textit{performance gap} between a given $\heur()$ and and an alternative \opt() (where \opt() can be the optimal solution).
The performance of a heuristic measures its solution quality.
For example, one measure of performance in TE is the total flow that the heuristic admits, and for packet schedulers, it is average packet delay.
The performance gap is the difference in performance between $\heur()$ and \opt().
\metaopt also finds \textit{adversarial inputs} to $\heur()$~---~those that cause a significant performance gap.
\metaopt can analyze a broad set of well-defined heuristics (see~\secref{s:design}).
Those we described earlier fit into this set.

\subsection{Using \metaopt to Analyze Heuristics}
\label{s:using-meta-analyze}

{
\begin{table*}
	\centering
	\begin{footnotesize}
		\renewcommand{\arraystretch}{1.5}
		\begin{tabular}{m{1.8em}|p{6.2cm}p{3.8cm}p{5.1cm}}
			& \metaopt finds performance gaps & \metaopt helps prove properties & \metaopt helps modify heuristics 
			\\ \hline
			{{\sf TE}}
			& \dempin and \pop are 33.9\% and 20\% less efficient than optimal. \metaopt finds realistic demands with strong locality that produce the same gap.
			& ~--~
			& Modified-\dempin  (which we designed with {\metaopt}'s help) has an order of magnitude lower gap than \dempin. We cannot improve the gap by running \dempin and \pop in parallel.
			\\ 
			{\vbp}
			& \metaopt finds tighter performance gaps under realistic constraints. It also finds the same examples which took theoreticians decades to find and prove a tight bound for 1d-\ffd.
			& \metaopt helped prove a tighter approximation ratio for 2-d \ffdsum than previously known~\cite{VBP-Heuristics}.
			& ~--~
			\\ 
			\pifo
			& \metaopt finds \sppifo can delay the highest priority packet by $3\times$, and also finds inputs where \aifo incurs $6\times$ more priority inversion than \sppifo.
			& \metaopt helped prove a new bound on the weighted average delay of \sppifo.
			& Modified-\sppifo (which \metaopt helped design) has $2.5\times$ lower weighted average packet delay.
			\\ 
			
		\end{tabular}
	\end{footnotesize}
	\caption{\metaopt (1) finds the performance gap between the heuristic and optimal; (2) helps prove various properties about the heuristic; and (3) helps modify them to improve their performance.\label{t:metaopt-results}
      }
\end{table*}
}

Heuristic designers can use \metaopt to answer the questions we presented in \secref{s:heur-their-import}.
They can use \metaopt in two ways: (a) to find performance gaps, and (b) to prove properties or improve heuristics based on adversarial inputs it finds.
We next describe the results we obtained when we applied \metaopt to the heuristics in \secref{s:heur-their-import}.
We present more detail and other results in \secref{s:eval}.
These use cases are not the only way operators can use \metaopt, but they demonstrate its versatility (see \secref{s:conc}).


\parab{Finding performance gaps.} We show how \metaopt helps find performance gaps in TE and packet scheduling.
%

\parae{Performance gaps in traffic engineering.}
We used \metaopt to find the performance gaps for \dempin and \pop where \opt() is the optimal multi-commodity flow algorithm.
We compute the performance gap~---~the difference between the heuristic and the optimal which we normalize by the total network capacity.
%
The performance gap is a lower bound on the \textit{optimality gap}, the worst-case gap between the two.

We find \dempin and \pop incur $33.9\%$ and $20\%$ relative performance gaps on a large topology (Cogentco, \secref{s:eval}).
This means there exists (and we can find) adversarial traffic demands that cause \dempin to use at least $33.9\%$ more capacity compared to optimal
%
Network operators that use \dempin may need to over-provision the network by that much to satisfy this demand.

\metaopt, by default, searches for adversarial demands among all possible demands.
We can constrain \metaopt to search over \textit{realistic} demands. 
These exhibit temporal locality where few node pairs communicate. The gap for \pop and \dempin reduces by less than 1\% when we run \metaopt with this constraint.
%

\parae{Performance gaps in packet scheduling.} We compare \sppifo to PIFO.
%
%
We compute and compare the priority-weighted average packet delay (\secref{s:eval}) between the two algorithms which penalizes them if they increase the delay of high-priority packets.
\metaopt shows there exists an input packet sequence where \sppifo is 3$\times$ worse than PIFO.

We used \metaopt to compare \sppifo and \aifo~\cite{aifo} (two heuristics).
\aifo emulates PIFO through a single FIFO queue and replaces $\heur()$ in this scenario.
{\metaopt} finds inputs for which \aifo incurs 6$\times$ more priority inversions than \sppifo.
Such analyses can help designers weigh performance trade-offs against switch resource usage.

\parab{Proving properties of and improving heuristics.} We used \metaopt to find adversarial inputs for various heuristics and analyzed these inputs to 
prove performance bounds for these heuristics or to improve them.

\parae{Proving a new bound for vector bin-packing.} \vbp heuristics try to minimize the number of bins they use.
Theoreticians prove bounds on their approximation ratio: the worst-case ratio of the number of bins the heuristic uses compared to the optimal over any input.
%
Recent work~\cite{VBP-Heuristics} showed 2-dimensional \ffdsum asymptotically approaches an approximation ratio of 2 (where the optimal uses nearly infinite bins).
We proved (\S\ref{s:eval}) the approximation ratio is always at least 2~---~even when the optimal requires a finite number of bins!

\parae{Proving a new bound for packet scheduling and improving the heuristic.}
We analyzed adversarial inputs \metaopt found for \sppifo and proved a lower bound on its priority-weighted average delay relative to PIFO. The bound is a function of the priority range and the number of packets.

Adversarial inputs to \sppifo trigger priority inversions, which means they cause \sppifo to enqueue high-priority packets behind low-priority ones.
We tested a Modified-\sppifo that splits queues into groups; it assigns each group a priority range, and runs \sppifo on each group independently.
Modified-\sppifo reduces the performance gap of \sppifo by {$2.5\times$}.

\parae{Improving traffic engineering heuristics.} We found \dempin 
 performs badly when small demands traverse long paths.
We used \metaopt to analyze a Modified-{\dempin} which only routes demands through their shortest paths if the shortest path is less than $k$ hops and the demand is less than than $T_{d}$.
This simple change reduced the performance gap an order of magnitude.

\section{\metaopt Design}
\label{s:design}

Our goal is to build a {\em widely applicable} system that finds {\em large performance gaps} between \opt and \heur~{\em quickly and at scale}.

\subsection{\metaopt Approach}
\label{s:metaopt-approach}

This is a hard problem when \opt and \heur are arbitrary algorithms,
but we observe we can formulate many heuristics in networks and systems as:

\begin{description}[nosep,leftmargin=*]
\item[Convex optimization problems] in which the heuristic seeks to optimize an objective, subject to a collection of constraints. \dempin falls into this category: for large demands it solves SWAN's optimization~\cite{swan}.
\item[Feasibility problems] in which the heuristic searches for a solution that satisfies a collection of constraints. \ffdsum falls into this category: it packs balls into bins subject to weight constraints.
\end{description}

%
%
%
%
%
%
%
%

\begin{figure}[t!]
\includegraphics[width=\linewidth]{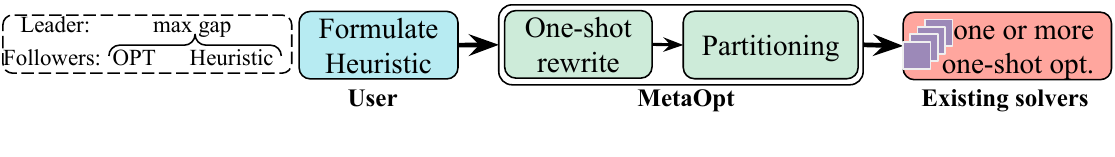}
\caption{{\metaopt}'s workflow involves four steps; (1) User encodes the heuristic~\secref{s:d:heur_form}, (2) {\metaopt} automatically applies rewrites to obtain a single-level optimization~\secref{s:automatic-rewrites}, (3) it partitions the problem into smaller subproblems to achieve scalability~\secref{s:d:part}, (4) {\metaopt} feeds the resulting optimizations into existing solvers~\cite{gurobi,Z3} and finds large performance gaps. \label{f:metaopt_workflow}}
\end{figure}

\begin{table}[t!]
	\begin{footnotesize}
		\begin{tabular}{cl|ll}
			 {\bf Optimal} & {\bf Heur.} & {\bf Formul.} & {\bf Rewrite} \\\hline
			 \multirow{2}{*}{\sf OptMaxFlow} & {\pop} & Convex(Random)~\secref{ss:convexify:te_1} & KKT/PD \\
			 & {\dempin} & Convex(Conditional)~\secref{ss:convexify:te_1} & KKT/PD  \\
			 {VBP MILP} & {\ffd} & NonConv.(Greedy)~\secref{ss:convexify:vbp_1} &  Feasibility \\
			 \multirow{2}{*}{\pifo MILP} & \sppifo  & NonConv.(Priority)~\secref{a:s:sp-pifo} & Feasibility \\
			 & \aifo & NonConv.(Admission)~\secref{a:s:aifo} & Feasibility
			 \\\hline
		\end{tabular}
	\end{footnotesize}
	\caption{Overview of the five heuristics we explored in this paper. We cover how to formulate the heuristics as optimizations in appendix and discuss their rewrites as constraints in~\secref{s:automatic-rewrites}.\label{t:metaopt_usecases}}
\end{table}

We can find the performance gap between any \opt and \heur through a bi-level or meta optimization (hence \metaopt) as long as \heur falls in one of these two classes. We do not need convexity for \opt~---~we need it to be either an optimization or a feasibility problem. We model the problem as\footnote{Note that we can transform minimization optimizations to maximization by negating their objective.}:

{\small
	\vspace{-3mm}
	\begin{alignat}{3}
&		\arg\max_{\mathcal{I}} \,\, {\sf H'}(\mathcal{I}) - {\sf H}(\mathcal{I}) && \quad\text{(leader problem)}\label{eq:metaopt_full_1} \nonumber\\
&		\mbox{s.t.} \quad {\mathcal{I} \in {\sf ConstrainedSet}} && \quad\text{(input constraints)}\nonumber\\
&		{\sf H}'(\mathcal{I}) = \max_{\mathbf{f}' \in \mathcal{F}'}{\sf H'\_Objective}(\mathbf{f}',\mathcal{I})  && \quad\text{(optimal)} \nonumber\\
&		{\sf H}(\mathcal{I}) = \max_{\mathbf{f} \in \mathcal{F}} {\sf H\_Objective}(\mathbf{f}, \mathcal{I})&& \quad\text{(heuristic)} 
\end{alignat}} 
\vspace{-2mm}
	
\noindent 
where the {\em leader} or {\em outer} optimization maximizes the difference between the two functions (\ie the performance gap) over a space of possible inputs $\mathcal{I}$.
This leader problem is subject to {\em follower} or {\em inner} problems (\opt and $\heur$).
We model the performance of ${\sf H'}$ when we apply it to input $\mathcal{I}$ through ${\sf H'\_Objective}$. This function decides the values for the variables $\mathbf{f}'$, internally encodes problem constraints, and computes the overall performance of \opt. The function ${\sf H'\_Objective}$ treats the outer problem's variables,  $\mathcal{I}$, as input and constant.
%
%
We define ${\sf H\_Objective}$  the same way.

%

\parab{Bi-level Optimization: A Brief Primer.}
Modern solvers do not directly support the style of bi-level optimizations we described in \eqnref{eq:metaopt_full_1}~\cite{GentleBilevel}.
To solve these, users usually rewrite the bi-level optimization as a single-level optimization~\cite{gurobi}.
These rewrites convert an optimization problem into a set of feasibility constraints:
if the inner problems are both optimizations, the rewrites would replace both ${\sf H'\_Objective}$ and ${\sf H\_Objective}$ with a set of feasibility constraints in the outer optimization.
The resulting formulation is a single-level optimization that modern solvers~\cite{gurobi,zen} can attempt to solve. 

\begin{figure}[t!]
  \begin{center}
    \resizebox{0.8\columnwidth}{!}{
    \begin{tikzpicture}[->,>=stealth]
      \node[state,
      text width=3.3cm] (A) 
      {\begin{tabular}{l}
         \textbf{min}$_{w, \ell}$ $\big( w^2 + \ell^2 \big)$\\[.5em]
         $2 \cdot (w + \ell) \geq P$ \\
         $P \geq 0$ \\
       \end{tabular}};
     
     \node[state,
     text width=3.2cm,
     right of=A,
     yshift = -1.1cm,
     node distance=4.8cm,
     anchor=center] (B)
     {%
       {\begin{tabular}{l}
          \textbf{solve} for $w, \ell, \lambda$\\[.5em]
          $2 (w + \ell) \geq P$ \\
          $P \geq 0$ \\
          $\lambda \geq 0$ \\
          $2w - 2\lambda = 0$ \\
          $2\ell - 2\lambda = 0$ \\
				\fcolorbox{myp}{myp}{$\lambda\big( w + \ell - \frac{P}{2} \big) = 0$}\\
				\end{tabular}}
		};
		
		\node[state,
		below of=A,
		yshift=-1.4cm,
		fill={rgb,255:red,222;green,226;blue,255},
		anchor=center,
		text width=3.6cm] (C) 
		{%
			{\renewcommand{\arraystretch}{1.3}
				\begin{tabular}{p{1.8cm} p{1.2cm}}
				$w = \ell = \frac{P}{4}$  & $\lambda = \frac{P}{4}$ \\
				\end{tabular}}
		};
		
		\path (A) 	edge[bend left=20]  node[anchor=south,above]{KKT}
		node[anchor=north,below]{encode} (B)
		(A)     	edge[bend right=20] node[anchor=south,right]{optimize}   (C)
		(B)       	edge[bend left=20]  node[anchor=south,above]{solve}      (C);
		
              \end{tikzpicture}
              }
	\end{center}
	\caption{{\label{f:eg_conversion_1} Example rewrite using KKT in an example where we find the optimal width $w$ and length $\ell$ for a rectangle such that the perimeter is $\geq P$. The inner variables are $w$ and $\ell$. The right panel shows the feasibility problem using the KKT theorem. Equations with $\lambda$ variables correspond to first order derivatives of inequality constraints in the original problem. $P$ is a variable of the outer optimization but is treated as a constant in the inner problem.}}
\end{figure}
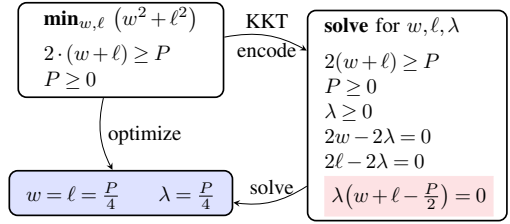

\figref{f:eg_conversion_1} is an example of such a rewrite.
This somewhat contrived example finds the optimal length and width of a rectangle subject to some constraints.
The rewrite uses the Karush–Kuhn–Tucker (KKT) theorem~\cite{boyd_co} and converts convex optimizations that have at least one strictly feasible point~\cite{boyd_co} into feasibility problems.
The theorem states any point that solves the new problem matches the solution of the original.
We describe another technique in \secref{s:d:rewrites} which exploits the Primal-Dual theorem~\cite{boyd_co}.
Both of these rewriting techniques produce a single-level optimization equivalent to the bi-level formulation if the inner problems are convex~\cite{GentleBilevel}.
The problem \metaopt solves has properties that allow us to reduce the overhead of these rewrites (see~\S\ref{s:automatic-rewrites}). 

\subsection{\metaopt: The User View}
\label{s:metaopt:-user-view}
\label{s:d:heur_form}

\parab{Inputs.}
\metaopt could have asked the user to input the single-level formulation, but this is hard and error-prone. 
The re-written single-level formulation can have an \textit{order of magnitude} more constraints than the original bi-level formulation (\secref{s:eval}) and is hard to optimize for 
practitioners that do not understand bi-level optimizations.

%

\metaopt simply lets the user input \opt and \heur (\figref{f:metaopt_workflow}).
It then automatically produces a single-level formulation, optimizes the re-writes, feeds these into a solver of the user's choice (\metaopt currently supports Gurobi~\cite{gurobi} and Zen~\cite{zen}), and produces performance gaps and adversarial inputs.

\parab{How to specify \opt or \heur}.
It can be hard to describe \opt or \heur in the language of optimization.
For example, a user who wants to compare the performance of \dempin and the optimal TE multi-commodity flow would have  to
specify the formulation of both of these algorithms.
Standard textbooks describe the former~\cite{BG}, so we focus on the latter.

\dempin involves a conditional (an \texttt{if} statement on the left of \figref{fig:dp-in-metaopt}).
We can use the big-M method~\cite{boyd_co} (\secref{sec:dempin_bigM}) to convert it into constraints optimization solvers support. 
\metaopt provides a helper function \textsf{ForceToZeroIfLeq} (see \figref{fig:dp-in-metaopt}) to help the user do this conversion.
This level of indirection not only makes it easier for the user to specify \dempin but also gives \metaopt the flexibility to optimize or change the formulations when it needs to.
For example, the big-M method causes numerical instability in larger problems and \metaopt uses an alternate method to convert it to constraints (see~\secref{ss:convexify:te_1}).

Our helper functions (\secref{s:a:helper}, \tabref{t:helper:list}) codify common design patterns and help specify constraints across a diverse set of problems.
%
We show  how to use and combine them to model other heuristics such as \ffd which involves greedy decisions, and \sppifo which involves dynamic weight updates.  
%
%
These helper functions make it easier for those without a background in optimization theory to use \metaopt and encode succinct and readable models. Users may still need to write some constraints themselves if we do not have a relevant helper function.
%

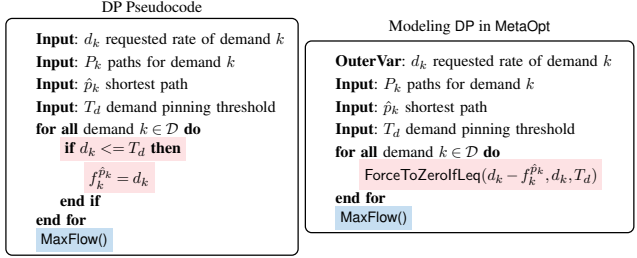
\begin{figure}[t!]
  \begin{center}
    \resizebox{\columnwidth}{!}{
		\begin{tikzpicture}[->,>=stealth]
			\node[state, text width=6.3cm, label=DP Pseudocode] (A)
			{
				\begin{algorithmic}[t]
					\Input : $d_k$ requested rate of demand $k$ \EndInput
					\Input : $P_k$ paths for demand $k$ \EndInput
					\Input : $\hat{p}_k$ shortest path \EndInput
					\Input : $T_d$ demand pinning threshold \EndInput
					\ForAll{{$\text{demand}~k \in \mathcal{D}$}}
					\CIF{$d_k <= T_d$}
					\State \colorbox{myp}{$f^{\hat{p}_k}_k = d_k$}
					\EndIf
					\EndFor
					\State \colorbox{myb}{{\sf MaxFlow()}}
				\end{algorithmic}
			};
			
			\node[state,
			text width=7.1cm,
			right of=A,
			node distance=7.0cm,
			anchor=center, label={Modeling {\sf DP} in {\sf MetaOpt}}] (B)
			{%
				
				\begin{algorithmic}[t]
					\OuterVar : $d_k$ requested rate of demand $k$ \EndOuterVar
					\Input : $P_k$ paths for demand $k$ \EndInput
					\Input : $\hat{p}_k$ shortest path \EndInput
					\Input : $T_d$ demand pinning threshold \EndInput
					\ForAll{{$\text{demand}~k \in \mathcal{D}$}}
					\State \colorbox{myp}{${\sf ForceToZeroIfLeq}(d_k - f^{\hat{p}_k}_k, d_k, T_d)$}
					\EndFor
					\State \colorbox{myb}{{\sf MaxFlow()}}
				\end{algorithmic}
			};

			
                      \end{tikzpicture}
                      }
              \end{center}
              \caption{The pseudo code for \dempin and how \metaopt models it.}
              \label{fig:dp-in-metaopt}
\end{figure}

\subsection{Automatic Rewrites}
\label{s:automatic-rewrites}

\metaopt produces a bi-level optimization from \opt and \heur (as in \eqnref{eq:metaopt_full_1}), and automatically rewrites it.
%
While the underlying theory behind rewrites is well-known~\cite{GentleBilevel}, to our knowledge, there are no automated re-writers.
Rewriting, till date, has required human intervention.
We need to be careful when we automate rewrites.
%
For example, it is hard to model non-linear constraints which include multiplication of variables. In a Primal-Dual rewrite (\secref{s:d:rewrites}), the constraints in the dual depend on whether the corresponding primal variable is unconstrained, if it is positive or negative, and on the type of optimization (maximization or minimization)~\cite{boyd_co}.
%

We have developed automatic rewrite techniques for KKT, Primal-Dual and a new variant of the latter (\ie Quantized Primal-Dual \secref{s:d:rewrites}).
Users can choose which rewrite they use.

%

\metaopt does not blindly re-write the bi-level formulation but only re-writes the inner problems if it needs to, a technique we call \textit{selective rewriting}.
In two instances, the inner problem does not require a re-write (\figref{fig::scope}): if it is a feasibility problem or if it is ``aligned'' (we can directly merge its constraints into the outer problem in both cases). 
An aligned inner problem is one where if we optimize for the outer problem's objective, we also (either directly or indirectly) optimize the inner problem as a consequence. 
For example, the outer problem is aligned with the inner problem if both problems are maximizations and when we maximize the objective of the outer problem, we also maximizes the objective of the inner problem.

We observe the objective of \metaopt is such that one of \opt or \heur is always aligned with the outer problem: the outer problem wants to maximize \opt and minimize \heur so that it maximizes the difference: this aligns with \opt if \opt is a maximization problem and with \heur when \heur is a minimization.
We do not need to rewrite the aligned inner problem and only need to merge its constraints into the outer problem and remove its objective (both \heur and \opt solve the same problem and therefore they are either both maximizations or both minimizations).
%

For all other instances, we need to re-write.
KKT and Primal-Dual rewrites only apply to models with linear or convex constraints (\sppifo and \ffdsum are non-convex examples).
\metaopt rewrites only if the inner problem is a \textit{non-aligned convex optimization}.

With all this, \metaopt generates a single-level formulation that scales well and preserves the theoretical properties that allow rewrites~---~the single-level formulation it produces is equivalent to the bi-level formulation.


\begin{figure}  
  \centering
  \resizebox{0.9\columnwidth}{!}{
  \begin{tikzpicture}
    \node[] at (-6,3)(First) {
      \begin{tikzpicture}[
        node distance=1.5cm,  
        auto,  
        every node/.style={align=center, font=\footnotesize},  
        decision/.style={diamond, fill=red!10, draw, text width=1.1cm, inner sep=0pt},  
        block1/.style={rectangle, fill=blue!10,  draw, text width=0.6cm, minimum height=0.5cm},  
        block2/.style={rectangle, fill=blue!10, draw, text width=1.4cm, minimum height=1cm},  
        block3/.style={rectangle, fill=blue!10, draw, text width=0.9cm, minimum height=0.5cm},  
        line/.style={draw, -latex}  
        ]  
        
        \node[decision] at (-2.3cm,0) (aligned) {aligned?};  
        \node[decision] at (-0.5cm,0) (feasibility) {feasibility?}; 
        \node[block1] at (-2.3cm, -1.4cm) (do_nothing1) {as is}; 
        \node[decision] at (1.2cm, 0) (convex) {convex?};
        \node[block1] at (-0.5cm, -1.4cm) (do_nothing2) {as is};
        \node[block2] at (3.1cm, 0) (kkt) {primal-dual\\kkt};
        \node[block3] at (1.2cm, -1.4cm) (do_nothing3) {NA};

        
        \draw[line, dashed] (aligned.west) ++ (-0.5cm,0cm) node[above, xshift=-0.2cm]{\opt or \heur} -- (aligned.west);

        \path[line] (aligned) -- node[above, text=red]{N} (feasibility); 
        \path[line] (aligned) -- node[left,text=ForestGreen]{Y} (do_nothing1);  
        \path[line] (feasibility) -- node[above, text=red]{N} (convex);
        \path[line] (feasibility) -- node[left,text=ForestGreen]{Y} (do_nothing2);
        \path[line] (convex) -- node[left, text=red]{N} (do_nothing3);
        \path[line] (convex) -- node[above,text=ForestGreen]{Y} (kkt);
      \end{tikzpicture}
    };

        
        
	
	
  \end{tikzpicture}
  }
  \caption{How \metaopt automatically converts the bi-level problem to a single-level optimization. \metaopt supports any follower which either (1) is a convex optimization; (2) is a feasibility problem; or (3) has an objective which aligns with the outer problem\label{fig::scope}.
  }  
\end{figure}
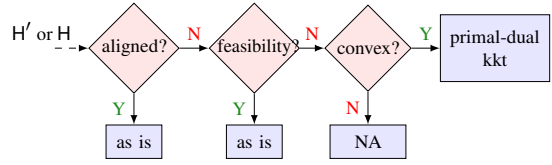

\vspace{0.05in}

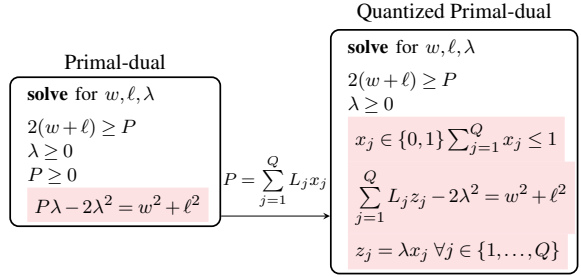
\begin{figure}[t!]
\begin{center}
	\resizebox{0.91\columnwidth}{!}{
	\begin{tikzpicture}[->,>=stealth]
		\node[state, text width=3.3cm, label=Primal-dual] (A)
		{
			{
				\small
				\begin{tabular}{l}
					\textbf{solve} for $w, \ell, \lambda$\\[.5em]
					{$2 (w + \ell) \geq P$} \\
					{$\lambda \geq 0$} \\
					$P \geq 0$ \\
					\fcolorbox{myp}{myp}{$P \lambda -2 \lambda^2 = w^2 + \ell^2$}
				\end{tabular}
			}
		};
		
		\node[state,
		text width=3.9cm,
		right of=A,
		node distance=5.55cm,
		anchor=center, label=Quantized Primal-dual] (B)
		{%
			{
				\small
				\begin{tabular}{l}
					\textbf{solve} for $w, \ell, \lambda$\\[.5em]
					{$2 (w + \ell) \geq P$} \\
					{$\lambda \geq 0$} \\
					\fcolorbox{myp}{myp}{$x_j \in \{0,1\} \sum_{j=1}^{Q} x_j \leq 1$}\\
					\fcolorbox{myp}{myp}{$\sum\limits_{j=1}^{Q}{L_j z_j} -2 \lambda^2 = w^2 + \ell^2$}\\
					\fcolorbox{myp}{myp}{$z_j = \lambda x_j \,\, \forall j \in \{1,\dots,Q\}$}
				\end{tabular}
			}
		};
		
		\begin{scope}[transform canvas={yshift=-3em}]
			\draw [->] (A) -- node[anchor=north,above] {{\footnotesize $P = \sum\limits_{j=1}^{Q}{L_j x_j}$}} (B);
		\end{scope}
	\end{tikzpicture}
	}
\end{center}
\caption{Left shows rewrite through primal-dual theorem. This rewrite has fewer constraints and different multiplicative terms~($\lambda P$ in the bottom equation versus $\lambda (w + \ell - \frac{P}{2})$ in the KKT rewrite). The quadratic constraint (in pink) is due to the quadratic objective of the original problem. On the right, we show how we quantize the parameters in the outer problem~($P$ in this case). The \qpd rewrite no longer has any multiplicative terms since we can linearize multiplication of binary and continuous variables. The values $L_j$ are constants and $x_j$ are binary variables.\label{f:pd_conversion_1}}
\end{figure}

\subsection{The Quantized Primal-Dual Rewrite}
\label{s:d:rewrites}

The KKT rewrite does not scale beyond small problem sizes.
This is primarily because it introduces multiplicative terms (pink highlighted constraint in \figref{f:eg_conversion_1}).
Commodity solvers support such constraints (\eg special ordered sets or SOS in Gurobi~\cite{gurobi} and disjunctions in Z3~\cite{Z3}).
However, these constraints slow down the solvers~--~the number of constraints with multiplications dictates the latency of the solver.

A similar observation holds for the Primal-Dual rewrite~\cite{GentleBilevel}.
This rewrite uses the strong duality theorem~\cite{boyd_co}: any feasible point of a convex problem is optimal iff the primal objective at that point is the same as the dual.
The primal-dual rewrite converts an optimization into a feasibility problem which contains the primal and dual constraints as well as an additional constraint that ensures the primal and dual objectives are equal.
In general, this rewrite can result in quadratic constraints; \figref{f:pd_conversion_1} (left) shows the primal-dual rewrite for the optimization we described in \figref{f:eg_conversion_1}.
Solvers like Gurobi~\cite{gurobi} support such constraints, but at the expense of scalability.

To scale, we have developed a technique, called \textit{Quantized Primal-Dual} (\qpd), to convert the primal-dual rewrite into an easier problem (see \figref{f:pd_conversion_1} for an end-to-end example).
%
%
%
%
We replace the input $P$ with $\sum_{j=1}^{Q} L_j x_j$ where $L_j$ are constants we choose a priori and $x_j$ are binary variables.
We require $\sum_j x_j \leq 1$.
This means the outer problem has to pick one of the $Q+1$ values ($0, L_1, \ldots, L_Q$) for $P$.
We only need to quantize the leader's variables that appear in the multiplicative terms in the primal-dual rewrite.
These variables only appear in the (misaligned) follower's constraints.
The inner problem is still optimal under this rewrite but we trade off optimality of the outer problem for speed by quantizing the input space.



Two challenges remain: (a) how to pick the number of quanta; and (b) which quanta to pick (\eg which values for $L_j$).
More quanta leads to more integer variables which slow down the solver.
However, fewer quanta leads to poor-quality adversaries because we now constrain the input variables to only a few pre-selected values.

Since the quantized primal-dual rewrite is much faster to solve than the other rewrites, we can sweep multiple quanta choices and pick the best.
We use the exact KKT rewrite on smaller problem instances to find good candidates. We observe empirically: adversarial inputs occur at so-called extreme points.
For example, the worst-case demands have values of either $0$ or the maximum possible in the \pop heuristic\footnote{For {\dempin}, the worst-case demands take values of $0$, the demand pinning threshold $T_d$, or the maximum possible demand.}.
Although we do not have a formal proof, we conjecture that the intuition behind these observations is similar to the intuition behind the simplex theorem~\cite{bertsimasLinOptBook}.
  
We evaluate how these rewrites impact the problem size and investigate the approximation gap quantized primal-dual introduces in \secref{s:eval}.
%

\begin{figure}[t!]
	\centering
	\subfigure[Clustered Topology.]{\centering\includegraphics[width=0.45\linewidth]{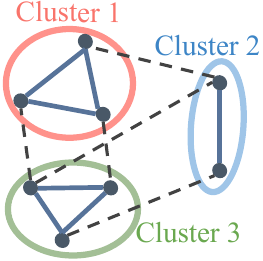}\label{fig:metaopt-cluster-example}} \hfill
	\subfigure[Clustered Demand Matrix.]{\centering\includegraphics[width=0.51\linewidth]{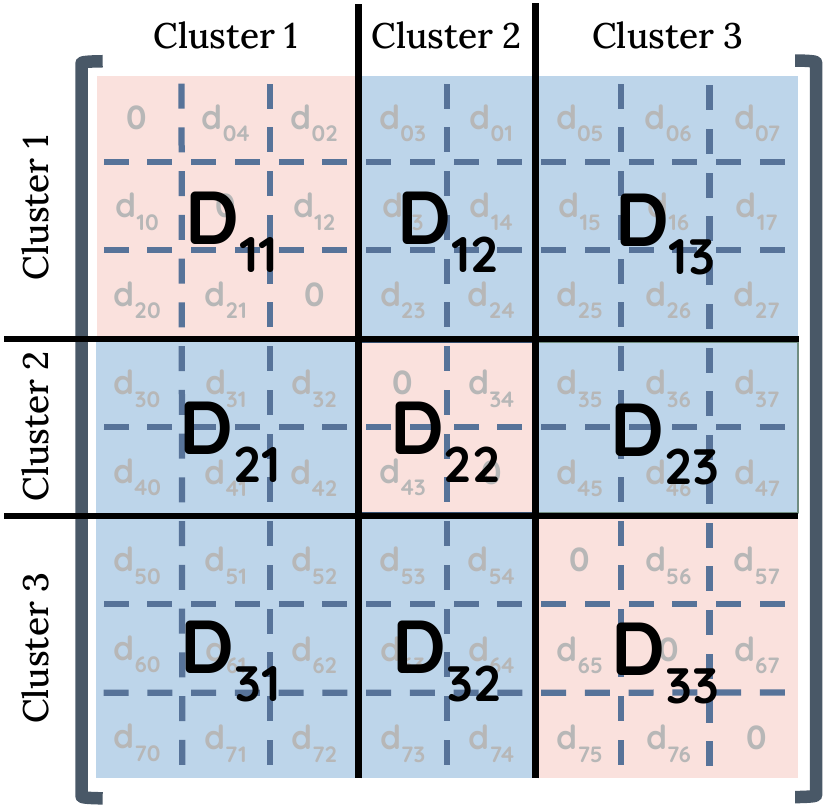} \label{fig:metaopt-cluster-demand-matrix}} \hfill
	\subfigure[MetaOpt Clustering Method.]{\centering\includegraphics[width=1.0\linewidth]{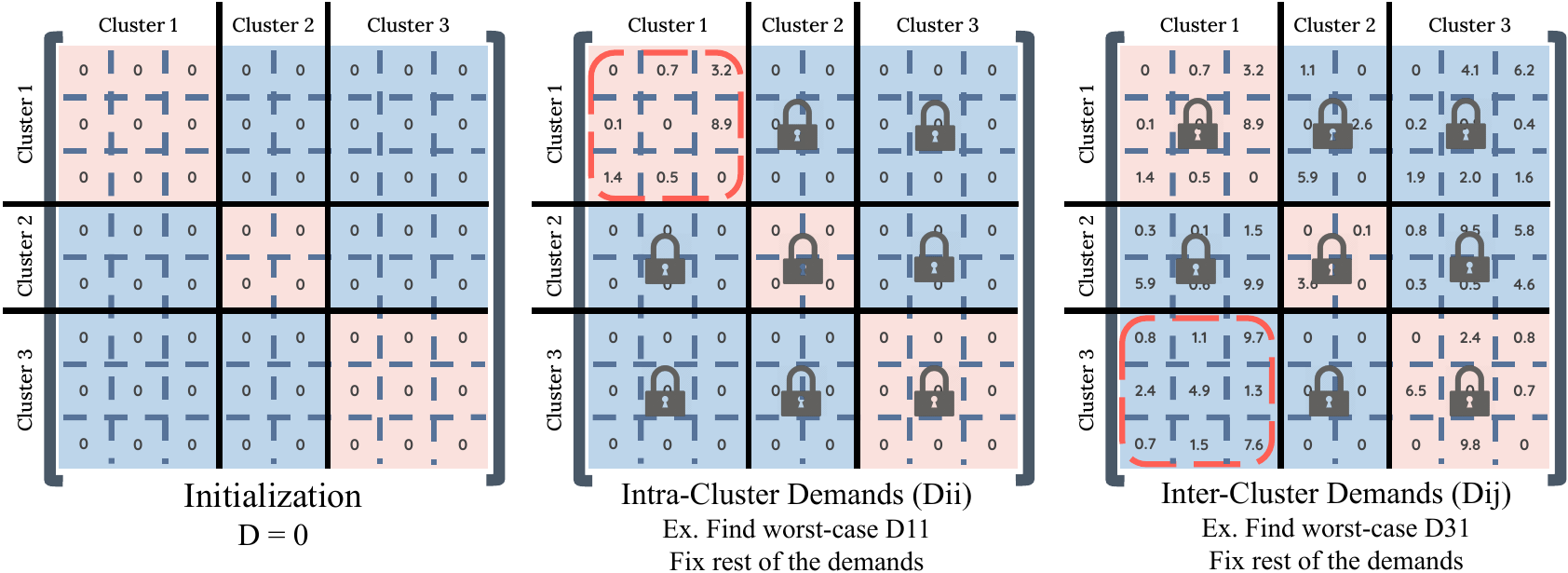} \label{fig:metaopt-cluster-method}}
	\caption{\label{f:partitioning} How we cluster in {\metaopt}. We first find demands that maximize the gap between \opt and \heur in each cluster. We then fix the demands within each cluster and one-by-one find demands between pairs of clusters that increase the gap (we fix the demands we found in the previous steps).}
\end{figure}

\subsection{Partitioning to Scale {\sf MetaOpt}}
\label{s:d:part}

To analyze problem instances at practical scales such as TE heuristics for realistic topologies and demands, we have found it necessary to use more aggressive scaling techniques.
One such technique is \textit{partitioning}.
We can partition any problem but we show the key steps in~\figref{f:partitioning} for the TE heuristics where the problem has an intrinsic graph structure (\ie the topology).
First, we partition nodes in the underlying network graph into clusters and solve the rewritten single-level optimization on each cluster in parallel.
In this step, we only consider the demands within each cluster~---~those on the block diagonal of the demand matrix (the pink blocks in~\figref{f:partitioning}).

We then freeze the demands from the last step and find demands between pairs of clusters that make the gap worse~---~we solve the rewritten problem on each pair of clusters.
This step iteratively fills the blue colored blocks in the demand matrix in~\figref{f:partitioning}.


We can parallelize the second step between cluster pairs with little overlap and produce an overall demand by adding the values \metaopt discovers after invoking each optimization.
This method speeds up \metaopt because each individual optimization, whether per cluster or per cluster pair, is much smaller than the overall problem.

We empirically find this partitioning approach consistently discovers inputs with large performance gaps.
This is because more than one adversarial input exists and our partitioning method does not bias against them.
For example, we find the adversarial inputs for \dempin follow a common pattern where demands between far apart nodes are just below the threshold. For such inputs, the heuristic wastes the capacity on many links when it routes the demands on their shortest paths.
In contrast, the optimal routing allocates those link capacities to multiple demands between nearby nodes.
Our partitioning method still allows \metaopt to find many inputs that have this pattern.

\section{Evaluation}
\label{s:eval}

We apply \metaopt to traffic engineering, vector bin packing, and packet scheduling heuristics to show it is general.
\metaopt helped us quantify and understand the performance gaps of heuristics, prove theoretical properties, and design better heuristics with lower performance gaps.
\tabref{t:metaopt-results} summarizes our findings.
We also show why the optimizations in \metaopt are important and quantify how quickly it finds performance gaps.

\parab{Implementation.}
Our \metaopt prototype is in C\# and uses Gurobi v9.5.2~\cite{gurobi}; we also have a port that uses Z3~\cite{Z3}.
To partition the graph, we adapt previous code~\cite{ncflow,communities_code} for spectral clustering~\cite{njw} and FM partitioning~\cite{louvain_method,communities} and report results for different cluster numbers and clustering techniques.

\subsection{Heuristics for WAN Traffic Engineering}
\label{s:eval:te}

In this section, we (a) obtain performance gaps for \dempin and \pop with respect to the optimal max-flow algorithm on large topologies, and  (b) devise modified versions of these heuristics based on our analysis of adversarial patterns.

\begin{table}[t]
	\centering
	{\footnotesize
		\begin{tabular}{c c c c c c}
			\textbf{Topology} & \textbf{\#Nodes} & \textbf{\#Edges} & \textbf{\#Part.} & \textbf{\dempin} & \textbf{\pop} \\\hline
			\textbf{\href{http://www.topology-zoo.org/maps/Cogentco.jpg}{Cogentco}} & $197$ & $486$ & $10$ & $33.9\%$ & $20.76\%$ \\
			\textbf{\href{http://www.topology-zoo.org/maps/Uninett2010.jpg}{Uninett2010}} & $74$ & $202$ & 8 & 28.4\% & 20.15\% \\
			\textbf{Abilene}~\cite{Abilene} & $10$ & $26$ & -- & $12.69\%$ & $17.31\%$ \\
			\textbf{B4}~\cite{b4} & $12$ & $38$ & -- & $13.16\%$ & $17.89\%$\\
			\textbf{SWAN}~\cite{swan} & $8$ & $24$ & -- & $2.29\%$ & $22.08\%$\\		
		\end{tabular}
	}
	\caption{\footnotesize Details of the topologies used in~\secref{s:eval:te} and discovered gap.\label{t:topos}}
\end{table}



\parab{Experiment Setup.}
We use $K$-shortest paths~\cite{YenKLoopLess} to find the paths between node pairs~($=4$ if unspecified).
We constrain the demands to be below a maximum value (half the average link capacity if unspecified) to ensure they are realistic and a single demand does not create a bottleneck.
For \dempin, we vary the demand pinning threshold (=$5$\% of average link capacity if unspecified).
For \pop, we vary the number of partitions (=$2$ if unspecified) and report the average gap over $5$ random trials (see~\secref{ss:convexify:te_1}).
We report runtimes on an AMD Operaton 2.4GHz CPU (6234) with 24 cores and 64GB of memory.
We use all available threads for all experiments (unless mentioned otherwise).
We timeout each optimization after $20$ minutes when we use partitioning.

\parae{Topologies.} We use two large topologies from~\cite{itzoo} and three public production topologies~\cite{swan,Abilene,b4} (\tabref{t:topos}).

\parae{Metrics.} We normalize the performance gap by the sum of the link capacities so we can compare across different scales.



\parab{Finding performance gaps.}
We compare \dempin and \pop to the optimal max-flow.
We use the Quantized Primal-Dual rewrite (\secref{s:d:rewrites}) and partitioning (\secref{s:d:part})
for most experiments but did not need the partitioning technique for small topologies (SWAN, B4 and Abilene).
We evaluate on topologies that range from 8 to nearly 200 nodes.
The performance gap, relative to the optimal, for \dempin ranges from $\sim2\%$ on SWAN to over 33\% (\tabref{t:topos}).
\pop also exhibits a large performance gap (up to 22\%).
%

These performance gaps are over any possible input demand.
We can also use \metaopt to obtain performance gaps on realistic inputs.
%
Production demands have strong locality~\cite{ncflow} and are very sparse~--~a small fraction of pairs have non-zero demands.
We can express these properties in \metaopt through constraints on the input space ($\mathcal{I}$ in \eqnref{eq:metaopt_full_1}).
%
With such inputs, the gaps for \dempin and \pop are only <1\% lower than the unconstrained case, but the adversarial demands we find are sparser and more local (\figref{f:real-constr:locality}).

\parab{Designing better heuristics.}
We can use the adversarial inputs \metaopt finds to design new heuristics or explore whether we can combine heuristics to improve them.
We first describe how we identified patterns in the adversarial inputs \metaopt discovered for \dempin and \pop and then show how we used these patterns to improve these heuristics.

\parae{Adversarial input patterns for \dempin.}
Intuitively, the performance gap of \dempin increases as we increase its threshold since the heuristic forces more demands on their shortest path.
Yet, the gap grows faster on some topologies even though they have roughly the same number of nodes and edges and even the same diameter!

We used synthetic topologies to study \dempin.
To create each topology, we start with a ring graph and then connect each node to a varying number of its nearest neighbors.
%
The results (\figref{f:dp_path_length}) indicate that the performance gap grows with the (average) shortest path length (fewer connections across nearest neighbors = longer shortest paths).
Intuitively, if the shortest path lengths are longer on average, \dempin will use the capacity on more edges to route the small demands.
This reduces the available capacity to route rest of the demands.

\begin{figure}[t!]
	\centering
	\subfigure[Gap vs. the threshold value for {\dempin}.]{
		\centering
		\includegraphics[width=0.89\linewidth]{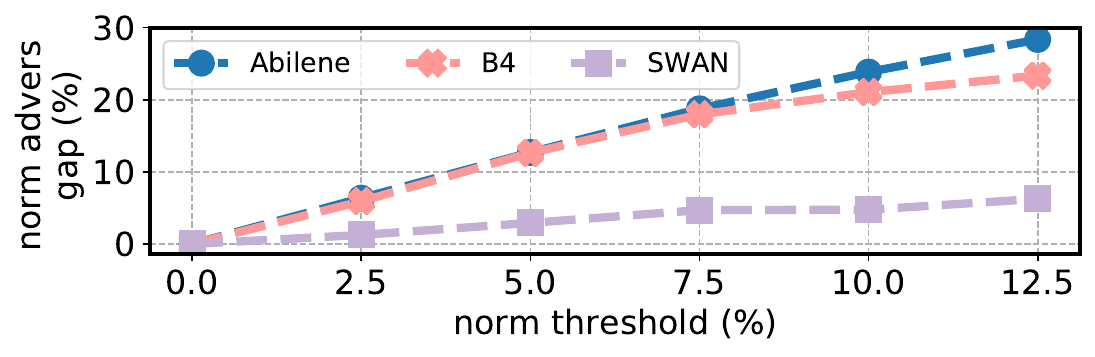}
		\label{f:dp_thresh}
	}
	\\ \vspace{-2mm}
	\subfigure[Gap vs. connectivity.]{
		\centering
		\includegraphics[width=0.89\linewidth]{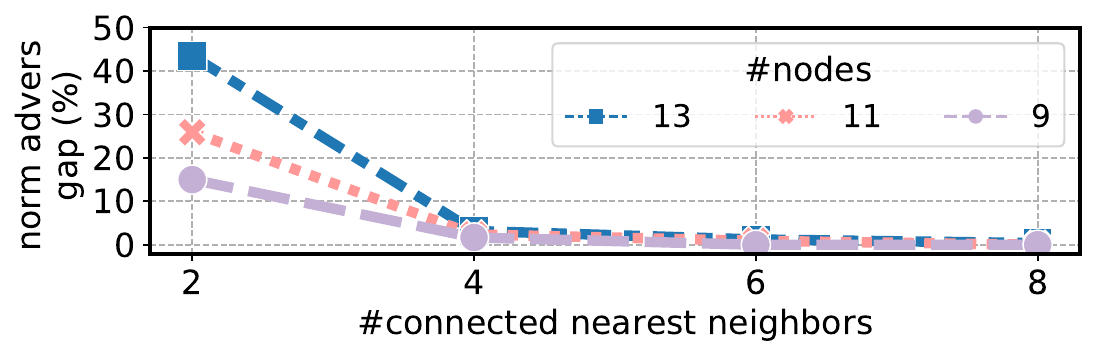}
		\label{f:dp_path_length}
	}
	\caption{\dempin's performance gap increases with the threshold and decreases with the connectivity.
		\label{f:maxgap_vs_params_dp}}
\end{figure}

\parae{Adversarial input patterns for \pop.}
Since \pop is a random heuristic, we search for inputs that maximize the expected gap.
We approximate this expectation by an empirical average over $n$ random partition samples.
Then, we check whether the adversarial inputs \metaopt finds can generalize by testing them on on 100 other random instances.
When \metaopt uses a small number of samples to estimate the expected gap, the adversarial inputs overfit.
We can improve its ability to generalize if we increase the number of samples but this hurts scalability.
%
We find $n=5$ is a sweet spot that permits scaling without overfitting.

As we increase the number of \pop partitions, the performance gap of \pop increases because each partition (1) gets a smaller fraction of each edge's capacity, and (2) has less information about the global state.
We can reduce this gap if we increase the number of paths as it helps each partition better allocate the fragmented capacity.

\begin{figure}[t!]
	\centering
	\subfigure[Gap vs. instances to approximate the expected value.]{\centering\includegraphics[width=0.89\linewidth]{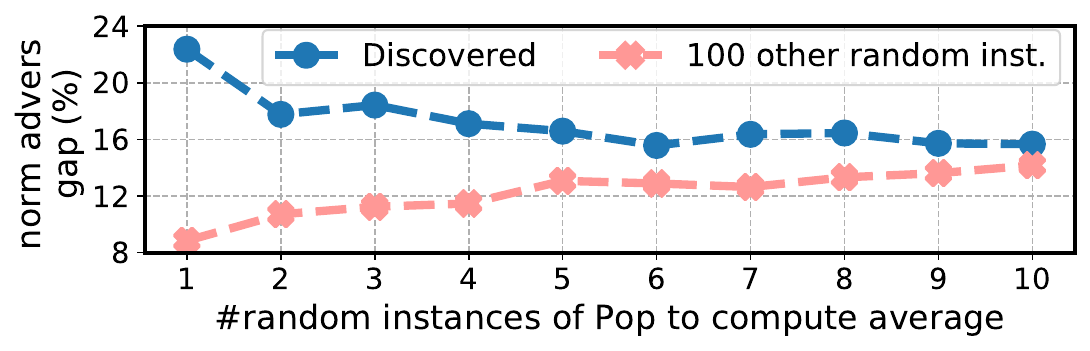}\label{f:expected_pop}}
	\\
	\subfigure[Gap vs. \#paths and \#partitions for avg {\pop}.]{\centering\includegraphics[width=0.89\linewidth]{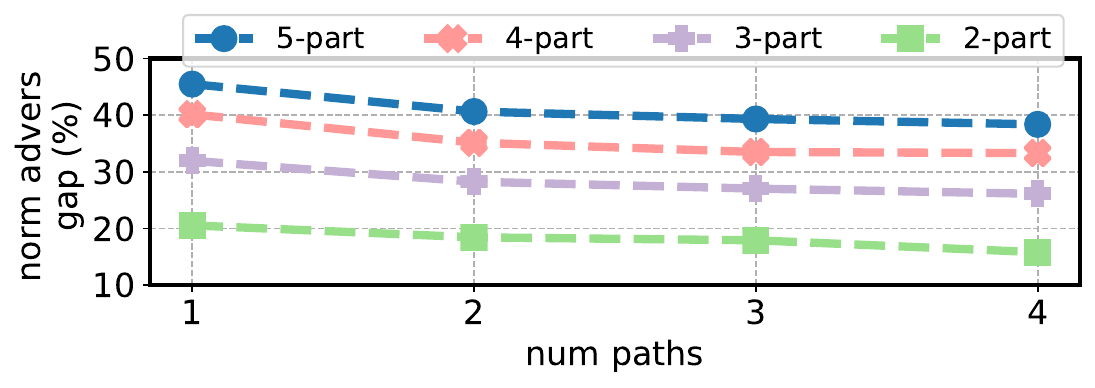}\label{f:pop_vs_numpartitions}}
	\caption{The performance gap of \pop on B4 when we vary: (a) \#instances to approximate average; and (b) \#paths and \#partitions.}
\end{figure}




\begin{figure}[t]
	\centering
	\subfigure[Impact of adding locality constraint on gap and density]{
		\centering
		{\footnotesize
			\begin{tabular}{c l c c }
				
				Heu & Constraint & Density & Gap \\ \hline
				\multirow{2}{*}{{\dempin}} & ~--~ & 54.06\% & 33.9\% \\ 
				& locality (distance of large demands $\leq 4$) & 12.03\% & 33.4\% \\\hline
				\multirow{2}{*}{{\pop}} & ~--~ & 16.14\% & 20.76\% \\ 
				& locality (distance of large demands $\leq 4$)& 4.74\% & 20.70\% \\ \hline
			\end{tabular}
		}
	}  
	\\
	\subfigure[Impact of adding locality constraint on {\dempin}]{
		\centering
		\includegraphics[width=0.89\linewidth]{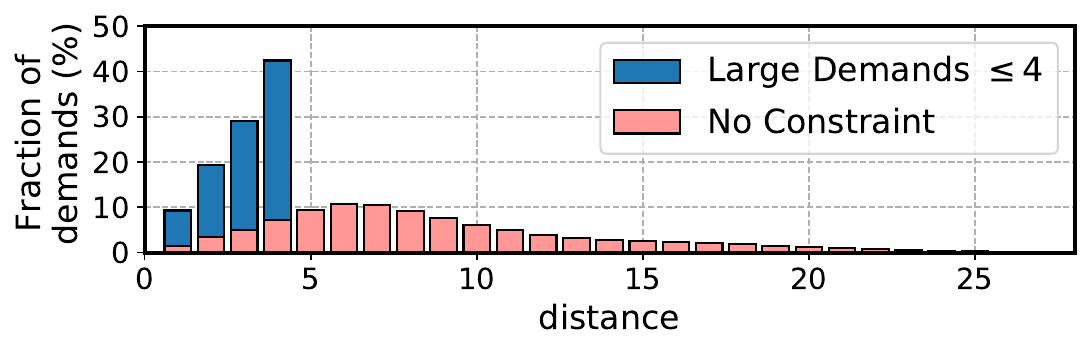}
	}
	\\
	\subfigure[Impact of adding locality constraint on {\pop}]{
		\centering
		\includegraphics[width=0.89\linewidth]{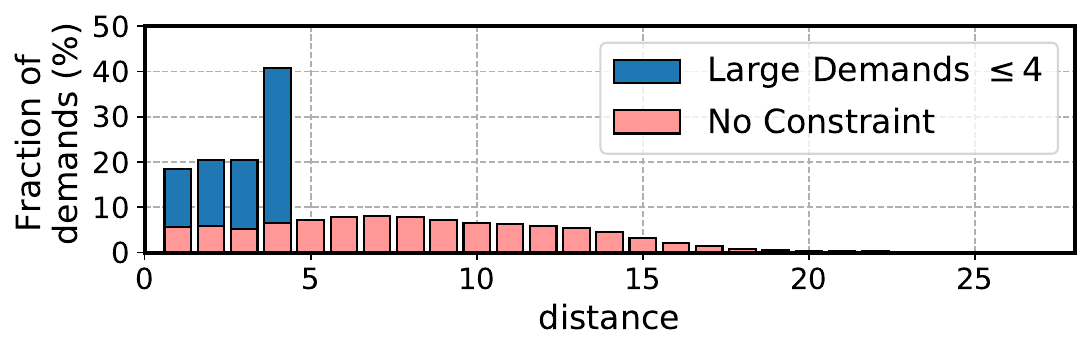}
	}
	\caption{We use \metaopt to find practical adversarial inputs on Cogentco. In both cases, the constraint on the input space led to more local and sparser adversarial inputs. \label{f:real-constr:locality}}
\end{figure}

\parae{Modified-\dempin.}
\metaopt shows \dempin has higher performance gap when nodes have larger average distance.
We use this insight to modify {\dempin} and design a better heuristic (modified-\dempin): we only pin a demand if it is below the threshold and {\em it is between nodes that are less than $k$ hops apart}.
The user specifies $k$.
The new heuristic routes small demands between far apart nodes optimally and leaves more capacity for other demands.
Its performance gap is $12.5\times$ smaller than \dempin for $T_d=1\%$ and $k = 4$.

As we increase the distance threshold, modified-\dempin can pin more demands (better scalability) but at the cost of higher performance gap.
{\metaopt} can help operators adjust the parameter $k$ based on their needs.

Modified-{\dempin} has another benefit: we can use a higher demand threshold $10$x-$50$x and maintain the same gap as the original {\dempin}.
We show this by fixing the gap at $5\%$ and computing the maximum threshold each method can admit using \metaopt (\figref{f:modified-dp:gap-fix}).
Operators can leverage this to pin more demands when small demands exhibit strong locality~\cite{ncflow}.

\begin{figure}[t]
	\centering
	\subfigure[Maximum threshold such that discovered gap $\sim 5\%$.]{
		\centering
		{\footnotesize
			\begin{tabular}{ c c c }
				
				Heuristic & Distance & Threshold wrt avg link cap \\ \hline
				{\dempin} & ~--~ & 0.1\% \\ \hline
				\multirow{2}{*}{modified-{\dempin}} & $\leq 6$ & $1\%$ (10x) \\
				& $\leq 4$ & $5\%$ (50x) \\ \hline
			\end{tabular}
		}
		\label{f:modified-dp:gap-fix}
	} 
	\\
	\subfigure[{\dempin} vs. modified-{\dempin}]{
		\centering
		\includegraphics[width=0.89\linewidth]{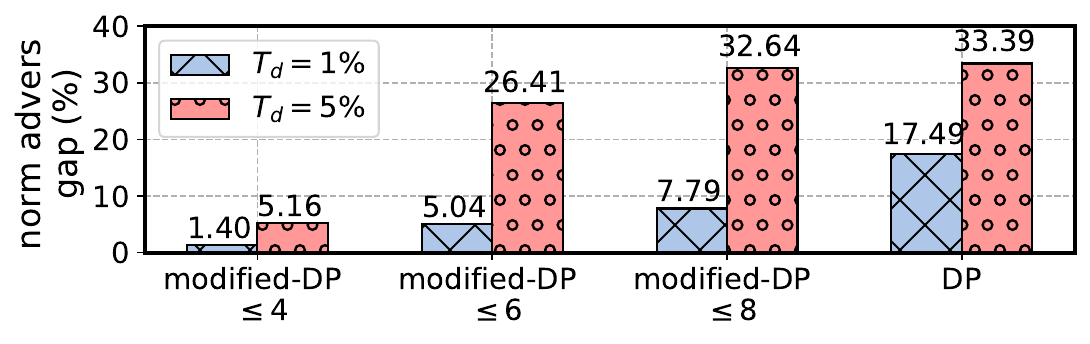}
		\label{f:modified-dp:thresh-fix}
	}
	\caption{We propose modified-{\dempin} based on insights from \metaopt. It only pins small demands between near nodes, is more resilient, and pins more demands with the same gap\label{f:modified-dp}.}
\end{figure}

\parae{Meta-\pop-\dempin.}
This meta-heuristic runs \pop and \dempin in parallel and selects the best solution for each input.
The two heuristics appear to have distinct adversarial inputs: \dempin under-performs when distant pairs have small demands and \pop when large demands that go through the same link end up in the same partition.
We expected combining them would reduce the performance gap significantly compared to each one.
But \metaopt shows the new heuristic only improves the performance gap by $6\%$ on the Cogentco topology.
It finds inputs where small demands are between distant pairs (adversarial to \dempin) and large demands are between nearby nodes that end up in the same partition (adversarial to \pop).

\subsection{Heuristics for Vector Bin Packing}
\label{s:eval:vbp}


We use \metaopt in vector bin packing (\vbp) to (a) derive performance gaps  that verify known results for \ffd in one dimension, and (b) prove a new property for \ffd in 2 dimensions.

\parab{Finding performance gaps.}
Performance in \ffd is measured by the number of bins it needs to fit a given number of balls and
the performance gap is a lower bound on the worst case approximation ratio (see~\S\ref{s:background}).
After decades of theoretical studies on 1d-\ffd~\cite{ffd-1d-tight,ffd-1d-1973,ffd-1d-1985,ffd-1d-1991}, the work in~\cite{ffd-1d-tight} established the tight bound $\ffd(\mathcal{I}) \leq \frac{11}{9} \eqopt(\mathcal{I}) + \frac{6}{9}$ for any $\mathcal{I}$ where \opt is the optimal.
To prove tightness, the authors craft a careful example where $\eqopt(\mathcal{I}) = 6$ bins and $\ffd(\mathcal{I})=8$.
\metaopt found the same example when we constrained its inputs to $\eqopt(\mathcal{I}) = 6$ and proved \ffd needs $8$ bins in the worst-case.

The proof in~\cite{ffd-1d-tight} assumes: (1) $\mathcal{I}$ can have an unlimited number of balls; and (2) the balls in $\mathcal{I}$ can have any size (even 0.00001$cm^3$!).
%
%
This is not always the case in practice: \eg when packing jobs (balls) in machines (bins)~\cite{jobpacking}, we often know \textit{apriori} an upper bound on the number of jobs or the quantization levels for resource requirements.
We can incorporate such constraints and ensure \metaopt finds practical performance gaps.
As \tabref{tab:1d-ffd} shows, when we constrain the number of balls and the ball sizes, \metaopt finds adversarial inputs that produce a tighter bound compared to~\cite{ffd-1d-tight}.



\begin{table}[t!]
	\centering
	{\footnotesize
		\begin{tabular}{c c c}
			max \#balls & ball size granularity & {\sf{FFD}}($\mathcal{I}_{{\metaopt}}$)\\ \hline
			20 & 0.01x &  8 \\
			20 & 0.05x &  7\\
			14 & 0.01x &  7\\
			\hline
		\end{tabular}
	}
	\caption{\metaopt finds slightly tighter bounds when constraining the number and size of balls. For \sf{OPT}($\mathcal{I}$)=\textbf{6}, the tightest known theoretical bound for \ffd~\cite{ffd-1d-tight} is \textbf{8}. This assumes the input can have an unlimited number of balls and that ball can have any size.\label{tab:1d-ffd}}
	\vspace{1mm}
	{\footnotesize
		\begin{tabular}{ c c c c c }
			
			\multirow{2}*{${\sf OPT}(\mathcal{I})$ }& \multicolumn{2}{c}{\metaopt} & \multicolumn{2}{c}{theoretical bound~\cite{VBP-Heuristics}} \\ \cline{2-5}
			& \#balls & approx ratio & \#balls & approx ratio \\ \hline
			2 & 6 & 2.0 & 4 & 1.0 \\
			3 & 9 & 2.0 & 12 & 1.33 \\
			4 & 12 & 2.0 & 24 & 1.5 \\
			5 & 15 & 2.0 & 40 & 1.6 \\
			\hline
		\end{tabular}
	}
	\caption{\metaopt finds adversarial examples with tighter approx. ratio for 2d-\ffd than the best known theoretical bound~\cite{VBP-Heuristics}.\label{tab:2d-ffd}}
		
\end{table}

%

\parab{Proving properties.}
While multi-dimensional \ffd is widely used in practice~\cite{protean,tetris,synergy}, its theoretical guarantees are less well understood.
Recently, \cite{VBP-Heuristics} crafted an example where 2-dimension \ffdsum uses $\alpha$ times more bins than the optimal, $\alpha \in [1,2)$, with $\alpha \to 2$ as the optimal tends to infinity.
In other words, for a finite size problem $\alpha$ is strictly less than 2.

\metaopt found adversarial inputs with $\alpha=2$ \emph{for every problem size we considered} (Table~\ref{tab:2d-ffd}).
For example, when the optimal uses 4 bins, \metaopt finds an adversarial input  with $12$ balls causing \ffdsum to use 8 bins.
In constrast, \cite{VBP-Heuristics} uses $24$ balls and only achieves approximation ratio of $1.5$.

The pattern we observe in adversarial inputs from \metaopt helped us prove the following theorem which establishes an approximation ratio of at least $2$ for \ffdsum. \secref{apendix:proof:ffdsum} contains the detailed proof.




\begin{theorem}
	\label{th:ffdsum:approx-ratio}
	In 2-dimensional \vbp, for any $k > 1$, there exists an input $\mathcal{I}$ with $OPT(\mathcal{I}) = k$ and $FFDSum(\mathcal{I}) \geq 2k$. 
\end{theorem}



\subsection{Heuristics for Packet Scheduling}
\label{s:packet-scheduling-eval}

We have used \metaopt to compare performance of \sppifo with both optimal (\pifo~\cite{pifo}) and another heuristic \aifo~\cite{aifo}.
Unlike \sppifo, \aifo uses a single queue but adds admission control based on packet ranks to approximate \pifo.

\begin{figure}[t]
	\includegraphics[width=0.91\linewidth]{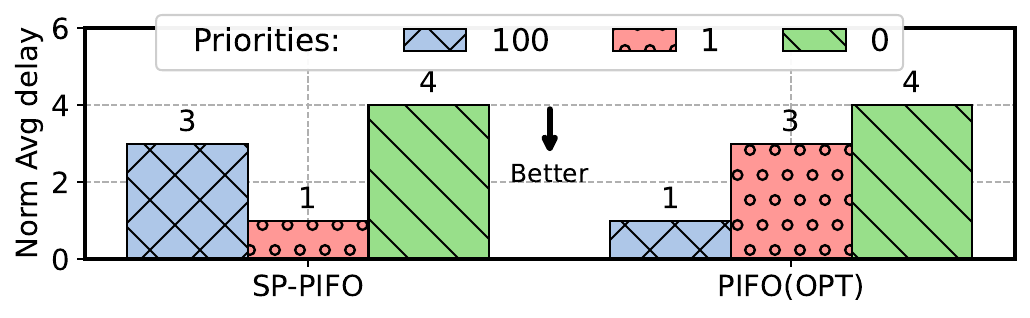}
	\caption{\sppifo can delay the highest priority packet (rank = 0) by $3\times$. Here, we show the average delay of packets with the same rank normalized by the average delay of the highest priority packets in \pifo. We assume packets have ranks between $0-100$ and the queues can admit all the packets (similar to \sppifo). For the case where $10K$ packets arrive at the same time and the queues drains at $40$ Gbps, the average delay for the highest priority packets in PIFO is $0.74ms$ (the performance gap in this figure is independent of the number of packets).\label{f:sp-pifo-pifo}}
\end{figure}

\parab{Finding performance gaps.}
We use the average delay of packets weighted by their priorities to compare \pifo and \sppifo.
\metaopt discovers packet traces where \sppifo fails to prioritize packets correctly and incurs $3\times$ higher delays for high priority packets (\figref{f:sp-pifo-pifo}) than \pifo!

We also use \metaopt to understand when, and by how much, \sppifo is better than \aifo, and when it is worse.
Unlike \sppifo, \aifo is designed for shallow-buffered switches and its admission control can drop packets.
For a fair comparison, we assume both heuristics use the same switch buffer size, and we split the buffer evenly across \sppifo queues.
With limited buffers, these algorithms may drop packets and we need to also compare the impact of their respective drop rates when we compare their performance impact on packets with different priorities.
%
We borrow a metric from \sppifo to do this: we count $k$ priority inversions when a packet is inserted in a queue with $k$ lower priority packets (even if the queue is full and the packet would have been dropped).
We found (\tabref{t:compare-sp-ai}):

\parae{\aifo sometimes outperforms \sppifo} because (1) it has one large queue instead of $n$ smaller ones~---~\sppifo drops many packets when faced with a burst of packets with the same priority because it assigns them to a single, smaller, queue;
and (2) \sppifo does not have admission control~---~we can create an adversarial pattern where we have lower priority packets arrive right before a group of high priority ones to make \sppifo admit the lower priority packets and drop the higher priority ones.

\begin{table}
	\centering
	\begin{footnotesize}
	\begin{tabular}{ccc}
		\multirow{2}{*}{\metaopt max objective} & \multicolumn{2}{c}{\#priority inversions} \\ \cline{2-3}
		& \sppifo~\cite{sp-pifo} & \aifo~\cite{aifo} \\ \hline
		$\aifo() - \sppifo()$ & 6 & 37 \\
		$\sppifo() - \aifo()$ & 24 & 11 \\\hline
	\end{tabular}
	\end{footnotesize}
	\caption{\metaopt discovered cases where one heuristic outperforms the other and vice versa. We show the number of priority inversions for \sppifo and \aifo on a trace from \metaopt with 18 packets. The total queue size is 12 and \sppifo has 4 queues.\label{t:compare-sp-ai}}
\end{table}

\parae{But \sppifo also sometimes outperforms \aifo} because (1) \aifo lacks a sorting mechanism which can cause high priority packets to get delayed behind lower priority ones, and (2) \aifo depends on an estimate of the distribution based on the most recent window of packets.
\metaopt found traces in which a few packets with completely different priorities compared to others can disrupt \aifo's distribution estimate!

\parab{Proving properties.}
{\metaopt} shows that adversarial inputs to \sppifo exhibit significant priority inversions.
Using this observation, we proved a lower-bound on the worst-case performance gap, in terms of priority-weighted average delay, between \sppifo and \pifo (see \secref{a:s:sp-pifo-theorem}):

\begin{theorem}
	\label{th:sp-pifo:avg-gap}
	For any number of packets $N$, integer priorities between $0-R_{max}$ and $q \ge 2$ queues, there exists a sequence of packets $\mathcal{I}$ where the
	difference between the weighted average packet delay that results from \sppifo is 
	\begin{equation}
		(R_{max} - 1)(N - 1 - p)p~~~~\text{where}~~p = {\left\lceil(N+1)/2\right\rceil}
	\end{equation}
	worse compared to PFIO.
\end{theorem}


\parab{Designing better heuristics.}
\sppifo uses the same set of queues to schedule packets with a wide range of priorities.
\metaopt found adversarial inputs in which a large number of low priority packets arrive before a high priority packet.
We find \sppifo under-performs when the difference between packet priorities is large.
\theoremref{th:sp-pifo:avg-gap} also confirms this, as the gap is proportional to $R_{max}$.

We evaluated a Modified-\sppifo, in which we reduced the range of packet priorities that can compete with each other:
we created $m$ queue groups where each group served a fixed priority range.
\sppifo runs on queues within a group.
This modified version can reduce the gap of \sppifo by $2.5\times$.



\subsection{Evaluating \metaopt}
\label{s:benchmarking-metaopt}

We show \metaopt can find solutions faster than other baseline search methods and helps users describe \opt and \heur in a compact way.
We also show how our various design choices help.

\parab{How fast \metaopt discovers a performance gap.}
Alternatives to \metaopt include random search, hill climbing~\cite{hill_climbing}, and simulated annealing~\cite{simulated_annealing}.
Random search repeatedly picks new random inputs and returns the one with maximum gap.
Hill climbing (HC) and simulated annealing (SA) use information from past observed inputs to guide the search (more details in~\xref{appendix:black-box}).

We compare them in \figref{f:method_gap_vs_latency:dp}.
{\metaopt} is the only method that consistently discovers substantially larger gaps over time, while other techniques get stuck in local optima and improve only slightly even after many hours (Fig.~\ref{f:dp-0.01-baseline:b4},~\ref{f:dp-0.01-baseline:Cogentco}).



\begin{figure*}[t!]
	\centering
	\subfigure[Gap vs. latency for B4 + {\dempin} ($T_d=1\%$)]{
		\centering
		\includegraphics[width=0.43\linewidth]{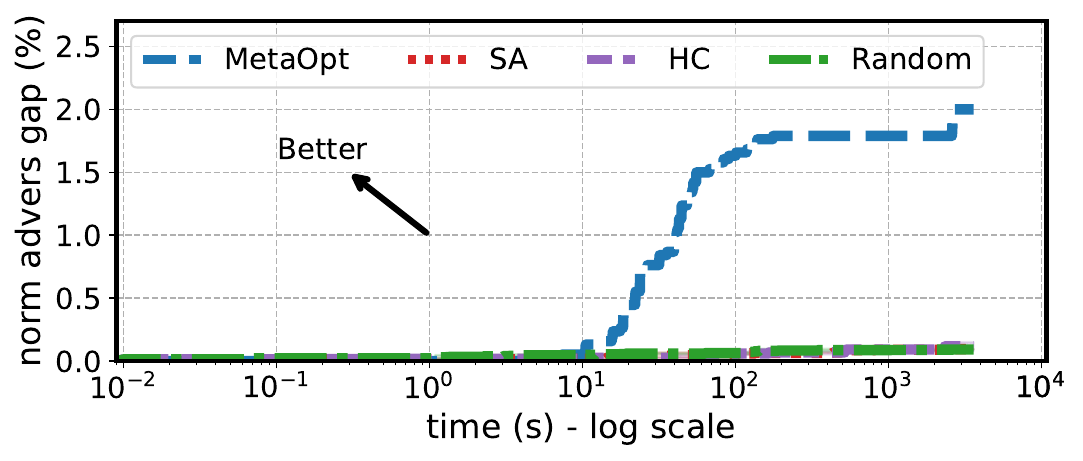}
		\label{f:dp-0.01-baseline:b4}
	}\hfill
	\subfigure[B4 + {\dempin} ($T_d=5\%$)]{
		\centering
		\includegraphics[width=0.215\linewidth]{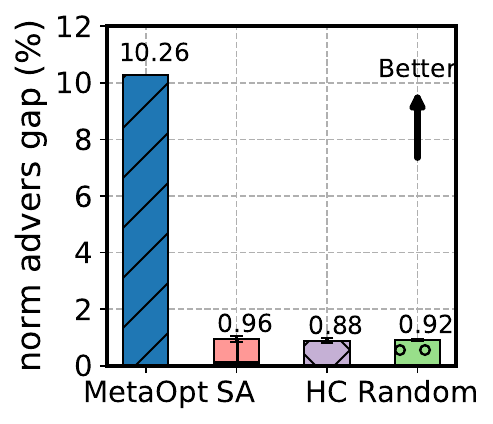}
		\label{f:dp-0.05-baseline:b4}
	}\hfill
	\subfigure[B4 + average gap of {\pop}]{
		\centering
		\includegraphics[width=0.215\linewidth]{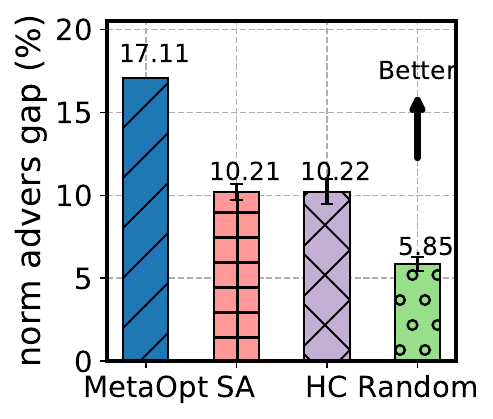}
		\label{f:pop_quality}
	}
	\\
	\subfigure[Gap vs. latency for Cogentco + {\dempin} ($T_d=1\%$)]{
		\centering
		\includegraphics[width=0.43\linewidth]{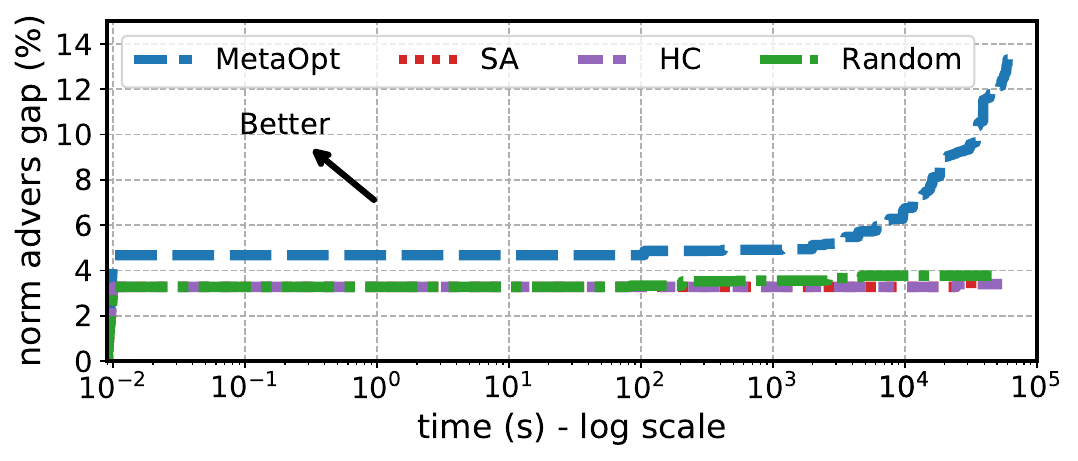}
		\label{f:dp-0.01-baseline:Cogentco}	
	} \hfill
	\subfigure[Cogentco + {\dempin} ($T_d=5\%$)]{
		\centering
		\includegraphics[width=0.215\linewidth]{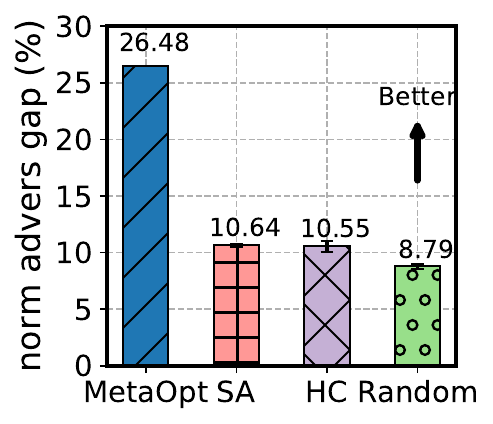}
		\label{f:dp-0.05-baseline:Cogentco}	
	} \hfill
	\subfigure[Cogentco + average gap of {\pop}]{
		\centering
		\includegraphics[width=0.215\linewidth]{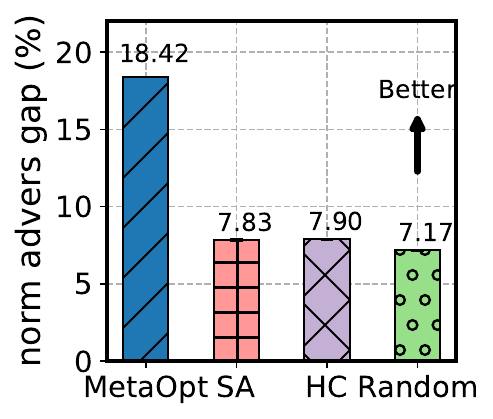}
		\label{f:pop-baseline:Cogentco}	
	}
	\caption{{\metaopt} is faster and finds larger gaps between {\sf OPT} and \pop or \dempin on the Cogentco and B4. For fair comparison, we report the gap relative to the total capacity and use only one thread to run each methods (SA = Simulated Annealing, HC = Hill Climbing).\label{f:method_gap_vs_latency:dp}}
\end{figure*}

\metaopt finds larger gaps ($1.7$x~-~$17$x compared to the next best).
The baselines fail since they ignore the details of the heuristic and treat it as black-box.
{\metaopt} uses its knowledge of the heuristic and the topology to guide the search.
%

\begin{figure}[t]
\includegraphics[width=0.9\columnwidth]{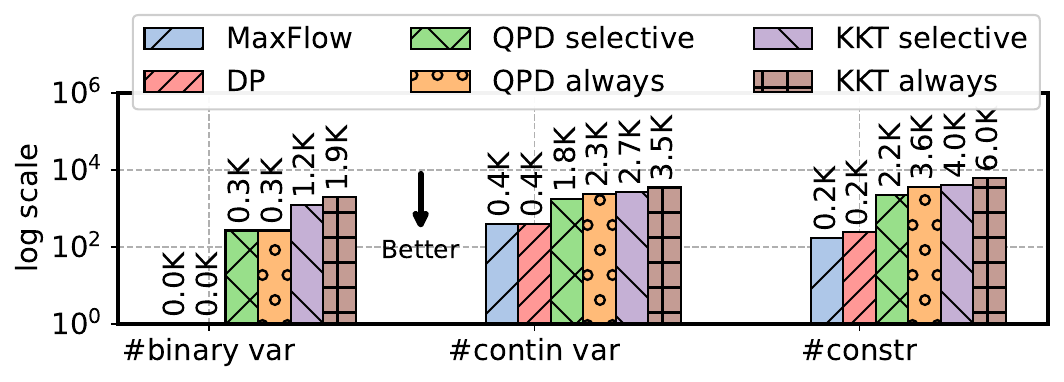}
\caption{Users specify \dempin and the optimal in \metaopt. We show how complex these specifications and the rewrites are in terms of the number of variables and constraints.}
\label{fig:complexity}
\end{figure}

\parab{Input and rewrite complexity.}
Users specify \opt and \heur in \metaopt and it automatically applies selective rewrites (\secref{s:automatic-rewrites}) to scale better.
We evaluate how complex these specifications and re-writes are in terms of the number of binary and continuous variables and the number of constraints they contain (\figref{fig:complexity}).
In general, these quantities impact the performance of a solver, and lower values improve solver performance.

We use \dempin as an example (results for \pop in \figref{fig:complexity:pop}) to highlight three features in {\metaopt}'s design (\figref{fig:complexity}).
%
The three metrics show the users inputs are more compact than the rewritten optimization; \eg they have a fifth of the constraints and half the number of continuous variables.
This quantifies how {\metaopt}'s automatic rewrites can reduce the programmer's burden.
Selective rewrites are important: we can use them to reduce the number of constraints (2.2K vs 3.6K for \qpd) and continuous variables (1.8K vs 2.3K for \qpd) compared to when we always rewrite the bi-level optimization.
%
We can produce more compact specifications (with less binary and continuous variables and fewer constraints) through \qpd compared to KKT even with selective rewrites.
This explains why it helps {\metaopt} scale.

\parab{The impact of partitioning.}
\metaopt partitions the problem to find larger gaps faster than both (non-partitioned) quantized primal-dual and KKT, even on medium-sized topologies (\figref{f:cluster-ablation:comparison}).
For larger topologies, neither KKT nor primal-dual can find large gaps without partitioning.

As we increase the number of partitions, \metaopt scales better and finds larger gaps until it eventually plateaus ($10$ for Cogentco in~\figref{f:cluster-ablation:num-clusters}).
We can slightly improve the gap if we double the solver timeout.

The inter-cluster step in partitioning is important, especially for heuristics that under-perform when demands are between distant nodes such as \dempin (\figref{f:cluster-ablation:inter-cluster}).
The partitioning algorithm also impacts the gap we find (\figref{f:cluster-ablation:graph}).

\begin{figure}[t!]
	\subfigure[KKT vs Quantized primal-dual vs w partitioning on Uninett2010]{
		\centering
		\includegraphics[width=0.91\linewidth]{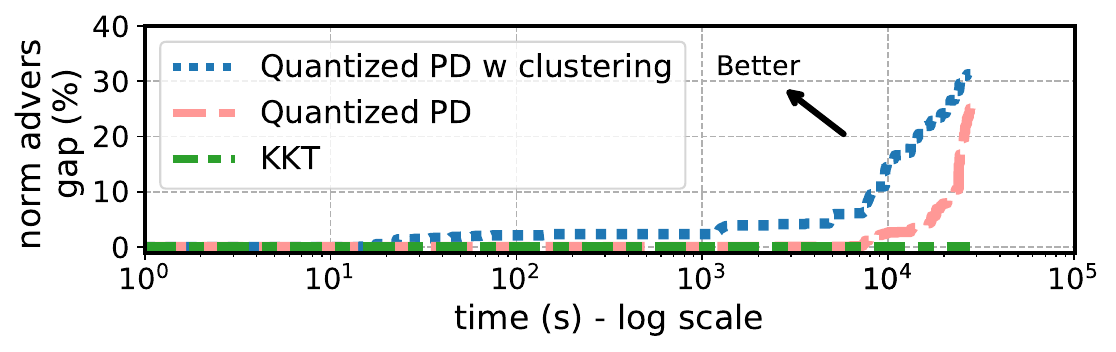}
		\label{f:cluster-ablation:comparison}
	}
	\subfigure[Number of partitions and solver timeout on Cogentco]{
		\centering
		\includegraphics[width=0.89\linewidth]{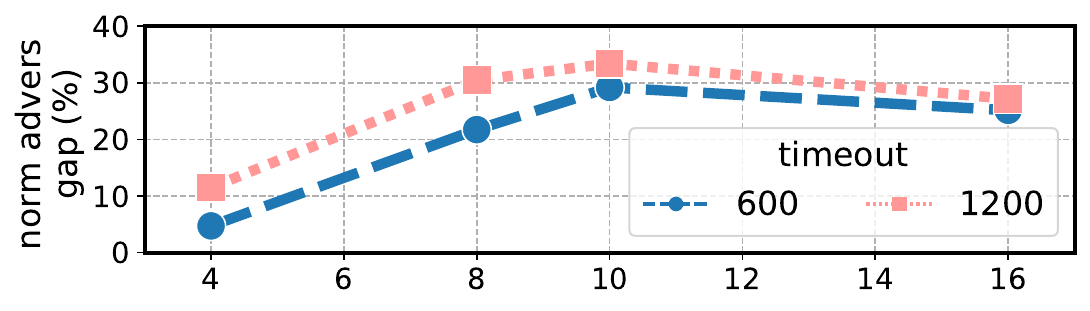}
		\label{f:cluster-ablation:num-clusters}
	}
	\subfigure[Inter-cluster on Cogentco]{
		\centering
		\includegraphics[width=0.43\linewidth]{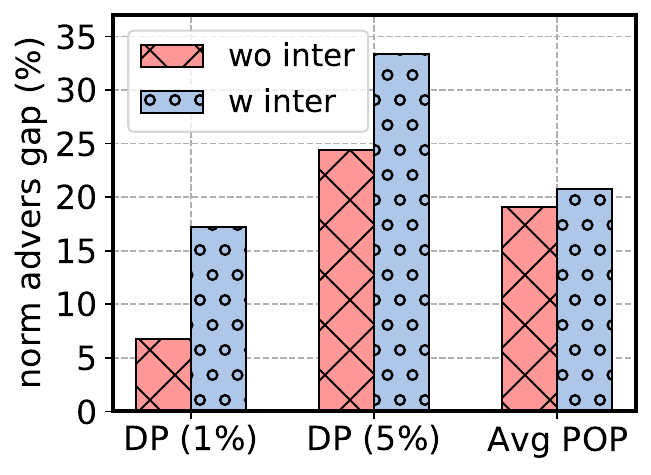}
		\label{f:cluster-ablation:inter-cluster}
	} \hfill
	\subfigure[Graph partitioning on Cogentco]{
		\centering
		\includegraphics[width=0.45\linewidth]{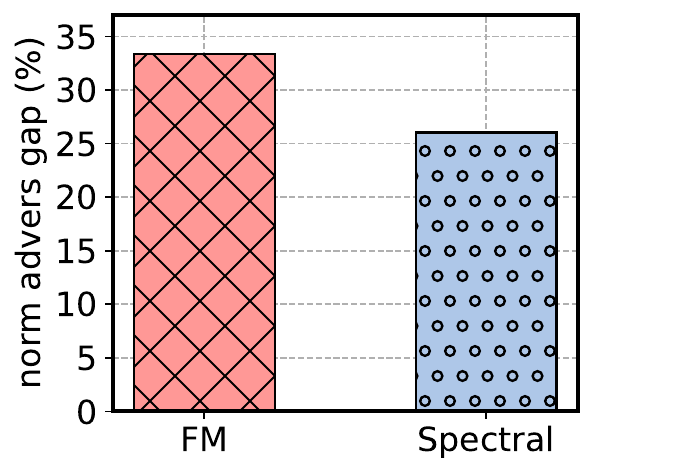}
		\label{f:cluster-ablation:graph}
	}
	\caption{We investigate the impact of partitioning in \metaopt's results. It helps \metaopt find larger gaps faster. \label{f:cluster-ablation}}
\end{figure}


				

			


\parab{The impact of quantization.}
%
%
To quantify, we compare the relative difference between the gap from quantized primal-dual and KKT (which does not use quantization). 
We found quantization has little impact on solution quality: $4\%$ for \dempin and 0 for \pop on B4 (we can not scale KKT to larger topologies).
For \pop, we use two quantiles: 0 and the max demand.
If a demand $d$ in an uncongested partition falls in between these values, forcing the demand to zero cannot decrease the gap: the rewrite's throughput would drop by $d$, and the optimal throughput by some value between 0 and $d$.
A similar argument applies to the congested case.

We use three quantiles for \dempin: 0, the threshold, and the max demand.
For a high enough threshold, quantized primal-dual may avoid assigning the threshold value to demands between distant nodes to not violate capacity constraints, whereas KKT can assign any value, which causes the relative difference in solution quality.

\section{How to extend \metaopt's scope}
\label{sec::stackelberg}

We can analyze any heuristic that we can specify as an optimization or feasibility problem through \metaopt.
%
To expand the scope of heuristics we can analyze, we observe the dynamics of the problem \metaopt addresses resembles leader-follower games (Stackelberg equilibria~\cite{sGames}).
In such games, a leader maximizes their payoff and controls the inputs to one or more followers.
With this {\em fixed} input from the leader, the followers have to respond and optimize for their own payoff and decide outputs, which the leader does not control but influences their payoff.

These games apply to a variety of leader-follower combinations (\eg optimization-based and Bayesian) and there are a variety of techniques to find the equilibrium in such games~\cite{zhang2021bayesian, karwowski2023sequential, wang2021fast, bosansky2015sequence}.
Hence, in theory, we can use these techniques to analyze the performance gaps of a broader class of heuristics than \metaopt currently supports, as long as we have a leader-follower combination where we know how to compute the equilibrium.
This is future work.

\vspace{-0.7mm}
\section{Related Work}
\label{s:related}

To the best of our knowledge, no prior work finds {\em provable} adversarial inputs for heuristics that approximate optimal problems. Our techniques (e.g., big-M, convex rewrites, and generally translating the problem to one that is amenable to off-the-shelf solvers) are not per-se novel~\cite{bilevel_3,bilevel_2,chemKKT} but no other work has combined them in this way. We also extend them to randomized, conditional, and sequential non-convex heuristics. Without our changes, we could not apply existing solvers directly or quickly find large gaps.

Our qualitative results -- the optimality gap and hard examples for {\pop}, {\dempin}, {\ffd}, {\sppifo}, and {\aifo}~--~are novel. We also find and prove tighter bounds for the optimality gap of \ffd and \sppifo.

Our work is different from most prior techniques. Traditional algorithmic worst- or average-case analysis~\cite{Cormen,randalgo} are specific to an individual heuristic and must be applied case by case. We cannot do such analysis for some heuristics as they only find loose bounds or do not account for realistic input constraints (none exist for \dempin, \pop~\cite{pop}, or the heuristics in~\cite{shoofly,arrow}). 

Verification methods also seek inputs that violate an invariant on a given function~\cite{software-model-checking}; but these methods support \emph{statically-specified} invariants on safety or correctness. In contrast, we seek inputs that maximize the gap between the optimal and heuristic algorithms. 

Model checking and approaches based on SMT solvers~\cite{CCAC, NetPerfFormalMethods, goel2023quantitative, cca-starvation, HeuristicSynthesis} can search for adversarial example inputs that result in performance gaps greater than a fixed bound when users can encode both the optimal and heuristic implementations as pure feasibility problems. However, these approaches cannot handle two-level optimization problems, where the optimal or heuristic solutions themselves must be formulated as optimizations (\eg traffic engineering). 


Local search algorithms~\cite{hill_climbing,simulated_annealing} apply to any (potentially black-box) heuristic or optimal algorithm but the flip side of such generality is they are slow on large input spaces, get stuck in local optima, and fail to find practical inputs because they ignore the inner workings of the heuristic.


Recent work find malicious inputs to learned techniques~\cite{rl_adversary,self_driving}. However, none of these works find provably large gaps or even consider the optimal algorithm.
Other broadly related work include: ~\cite{cuckoo_break_without_knowing_hash_functions,cyber_physical_testing,raregan, surgeprotector}. 

Our partitioning approach is different from that of
recent work~\cite{ncflow,pop}.
\metaopt does not need to return a {\em feasible flow allocation} (where flows respect the capacity and path constraints) for either \opt or \heur but instead only has to find adversarial inputs~---~it does not need to combine the solutions of the sub-problems.
These sub-problems often fail to enforce global constraints which is why NCFlow and \pop are more complex.
\pop~\cite{pop} sacrifices quality and partitions demands separately to ensure the sub-problems enforce rigorous constraints.
\metaopt does not need such constraints because the result it produces is a choice of inputs for the problem and not the solution.
%
Notice a certain `coming full circle' aspect here: to quickly analyze the optimality gap of {\pop}, we use a similar (but not the same) partitioning.

This paper is an extended version of~\cite{metaopt}.
Compared to this workshop paper, we changed the methodology to improve generality and scalability, added helper functions for ease of use, added support for heuristics from VBP and packet scheduling, and did a more extensive evaluation.

\vspace{-0.9mm}
\section{Summary and Future Work}
\label{s:conc}

\metaopt is a heuristic analyzer for heuristics that can be posed as an optimization or a feasibility problem.
It can be used to find performance gaps at scale, prove lower bounds on worst-case gaps, and devise improvements to heuristics.
At its core, \metaopt solves a bi-level optimization problem.
To do this efficiently, it selectively rewrites heuristic specifications as a single-level optimization, and incorporates several scaling techniques.
Future work can include using it to evaluate and improve other heuristics, increasing its expressivity (\secref{sec::stackelberg}), and identifying infeasible inputs instead of adversarial ones.
 

\clearpage
\bibliographystyle{plain}
	\bibliography{./bibs/wanrisk}

\begin{thebibliography}{10}

\bibitem{itzoo}
{Internet Topology Zoo}.
\newblock \url{http://www.topology-zoo.org/}.

\bibitem{yarn_vbp}
Yarn resource allocation of multiple resource-types.
\newblock \url{https://bit.ly/3YMDL2Z}.

\bibitem{ncflow}
Firas Abuzaid, Srikanth Kandula, Behnaz Arzani, Ishai Menache, Matei Zaharia,
  and Peter Bailis.
\newblock Contracting wide-area network topologies to solve flow problems
  quickly.
\newblock In {\em NSDI}, 2021.

\bibitem{HeuristicSynthesis}
Anup Agarwal, Venkat Arun, Devdeep Ray, Ruben Martins, and Srinivasan Seshan.
\newblock Automating network heuristic design and analysis.
\newblock In {\em HotNets}, 2022.

\bibitem{sp-pifo}
Albert~Gran Alcoz, Alexander Dietm{\"u}ller, and Laurent Vanbever.
\newblock $\{$SP-PIFO$\}$: Approximating $\{$Push-In$\}$$\{$First-Out$\}$
  behaviors using $\{$Strict-Priority$\}$ queues.
\newblock In {\em 17th USENIX Symposium on Networked Systems Design and
  Implementation (NSDI 20)}, pages 59--76, 2020.

\bibitem{NetPerfFormalMethods}
Mina~Tahmasbi Arashloo, Ryan Beckett, and Rachit Agarwal.
\newblock Formal methods for network performance analysis.
\newblock In {\em NSDI}, 2023.

\bibitem{bilevel_3}
Bryan Arguello, Richard~L. Chen, William~E. Hart, John~D. Siirola, and
  Jean-Paul Watson.
\newblock Modeling bilevel program in pyomo.
\newblock \url{https://www.osti.gov/servlets/purl/1526125}.

\bibitem{cca-starvation}
Venkat Arun, Mohammad Alizadeh, and Hari Balakrishnan.
\newblock Starvation in end-to-end congestion control.
\newblock In {\em SIGCOMM}, 2022.

\bibitem{CCAC}
Venkat Arun, Mina~Tahmasbi Arashloo, Ahmed Saeed, Mohammad Alizadeh, and Hari
  Balakrishnan.
\newblock Toward formally verifying congestion control behavior.
\newblock In {\em SIGCOMM}, 2021.

\bibitem{surgeprotector}
Nirav Atre, Hugo Sadok, Erica Chiang, Weina Wang, and Justine Sherry.
\newblock Surgeprotector: Mitigating temporal algorithmic complexity attacks
  using adversarial scheduling.
\newblock In {\em SIGCOMM}, 2022.

\bibitem{ffd-1d-1985}
Brenda~S Baker.
\newblock A new proof for the first-fit decreasing bin-packing algorithm.
\newblock {\em Journal of Algorithms}, 1985.

\bibitem{GentleBilevel}
Y.~Beck and M.~Schmidt.
\newblock {A Gentle and Incomplete Introduction to Bilevel Optimization}.
\newblock \url{https://optimization-online.org/?p=17182}, July 2023.

\bibitem{zen}
Ryan Beckett and Ratul Mahajan.
\newblock A general framework for compositional network modeling.
\newblock In {\em Proceedings of the 19th ACM Workshop on Hot Topics in
  Networks}, pages 8--15, 2020.

\bibitem{BG}
Dimitri Bertsekas and Robert Gallager.
\newblock {\em {Data Networks}}.
\newblock Englewood Cliffs, 1992.

\bibitem{bertsimasLinOptBook}
Dimitris Bertsimas and John~N Tsitsiklis.
\newblock {\em {Introduction to linear optimization}}, volume~6.
\newblock Athena Scientific Belmont, MA, 1997.

\bibitem{louvain_method}
Vincent~D Blondel, Jean-Loup Guillaume, Renaud Lambiotte, and Etienne Lefebvre.
\newblock {Fast unfolding of communities in large networks}, 2008.

\bibitem{teavar}
Jeremy Bogle, Nikhil Bhatia, Manya Ghobadi, Ishai Menache, Nikolaj Bj\o{}rner,
  Asaf Valadarsky, and Michael Schapira.
\newblock {TEAVAR: striking the right utilization-availability balance in WAN
  traffic engineering}.
\newblock In {\em SIGCOMM}, 2019.

\bibitem{bosansky2015sequence}
Branislav Bosansky and Jiri Cermak.
\newblock Sequence-form algorithm for computing stackelberg equilibria in
  extensive-form games.
\newblock In {\em Proceedings of the AAAI Conference on Artificial
  Intelligence}, volume~29, 2015.

\bibitem{boyd_co}
Stephen Boyd and Lieven Vandenberghe.
\newblock {\em {Convex Optimization}}.
\newblock Cambridge University Press, 2004.

\bibitem{games}
Wu~chang Feng, Francis Chang, Wu~chi Feng, and Jonathan Walpole.
\newblock Provisioning on-line games: A traffic analysis of a busy
  counter-strike server.

\bibitem{communities}
A.~Clauset, M.E.J. Newman, and C.~Moore.
\newblock {Finding community structure in very large networks}.
\newblock {\em Phys. Rev.}, 2004.

\bibitem{communities_code}
Aaron Clauset.
\newblock {Fast Modularity Community Structure Inference Algorithm}.
\newblock \url{https://bit.ly/3aAVGQH}.

\bibitem{bilevel_2}
Benoît Colson, Patrice Marcotte, and Gilles Savard.
\newblock Bilevel programming: A survey.
\newblock {\em 4OR}, 2005.

\bibitem{Cormen}
Thomas~H Cormen, Charles~E Leiserson, Ronald~L Rivest, and Clifford Stein.
\newblock {\em Introduction to algorithms}.
\newblock MIT press, 2022.

\bibitem{cyber_physical_testing}
Anthony Corso, Robert Moss, Mark Koren, Ritchie Lee, and Mykel Kochenderfer.
\newblock A survey of algorithms for black-box safety validation of
  cyber-physical systems.
\newblock {\em Journal of AI Research}, 2021.

\bibitem{hill_climbing}
L.~Davis.
\newblock Bit-climbing, representational bias, and test suit design.
\newblock {\em Proc. Intl. Conf. Genetic Algorithm}, pages 18--23, 1991.

\bibitem{ffd-1d-tight}
Gy{\"o}rgy D{\'o}sa.
\newblock The tight bound of first fit decreasing bin-packing algorithm is $ffd
  (i) \leq 11/9 opt (i) + 6/9$.
\newblock In {\em Combinatorics, Algorithms, Probabilistic and Experimental
  Methodologies: First International Symposium, ESCAPE}. Springer, 2007.

\bibitem{chemKKT}
Pablo Garcia-Herreros, Lei Zhang, Pratik Misra, Erdem Arslan, Sanjay Mehta, and
  Ignacio~E Grossmann.
\newblock Mixed-integer bilevel optimization for capacity planning with
  rational markets.
\newblock {\em Computers \& Chemical Engineering}, 86:33--47, 2016.

\bibitem{drf}
Ali Ghodsi, Matei Zaharia, Benjamin Hindman, Andy Konwinski, Scott Shenker, and
  Ion Stoica.
\newblock Dominant resource fairness: Fair allocation of multiple resource
  types.
\newblock In {\em NSDI}, 2011.

\bibitem{rl_adversary}
Tomer Gilad, Nathan~H. Jay, Michael Shnaiderman, Brighten Godfrey, and Michael
  Schapira.
\newblock Robustifying network protocols with adversarial examples.
\newblock In {\em HotNets}. ACM, 2019.

\bibitem{goel2023quantitative}
Saksham Goel, Benjamin Mikek, Jehad Aly, Venkat Arun, Ahmed Saeed, and Aditya
  Akella.
\newblock Quantitative verification of scheduling heuristics.
\newblock {\em arXiv preprint arXiv:2301.04205}, 2023.

\bibitem{tetris}
Robert Grandl, Ganesh Ananthanarayanan, Srikanth Kandula, Sriram Rao, and
  Aditya Akella.
\newblock Multi-resource packing for cluster schedulers.
\newblock In {\em SIGCOMM}, 2014.

\bibitem{gurobi}
{Gurobi Optimization, LLC}.
\newblock {Gurobi Optimizer Reference Manual}, 2022.

\bibitem{protean}
Ori Hadary, Luke Marshall, Ishai Menache, Abhisek Pan, Esaias~E Greeff, David
  Dion, Star Dorminey, Shailesh Joshi, Yang Chen, Mark Russinovich, and Thomas
  Moscibroda.
\newblock Protean: Vm allocation service at scale.
\newblock In {\em OSDI}, 2020.

\bibitem{swan}
Chi-Yao Hong, Srikanth Kandula, Ratul Mahajan, Ming Zhang, Vijay Gill, Mohan
  Nanduri, and Roger Wattenhofer.
\newblock {Achieving high utilization with software-driven WAN}.
\newblock In {\em SIGCOMM}, 2013.

\bibitem{b4}
Sushant Jain, Alok Kumar, Subhasree Mandal, Joon Ong, Leon Poutievski, Arjun
  Singh, Subbaiah Venkata, Jim Wanderer, Junlan Zhou, and Min Zhu.
\newblock {B4: Experience with a globally-deployed software defined WAN}.
\newblock In {\em SIGCOMM}, 2013.

\bibitem{pretium}
Virajith Jalaparti, Ivan Bliznets, Srikanth Kandula, Brendan Lucier, and Ishai
  Menache.
\newblock {Dynamic pricing and traffic engineering for timely inter-datacenter
  transfers}.
\newblock In {\em SIGCOMM}, 2016.

\bibitem{software-model-checking}
Ranjit Jhala and Rupak Majumdar.
\newblock Software model checking.
\newblock {\em ACM Comput. Surv.}, 41:21:1--21:54, 2009.

\bibitem{ffd-1d-1973}
David~S Johnson.
\newblock {\em Near-optimal bin packing algorithms}.
\newblock PhD thesis, Massachusetts Institute of Technology, 1973.

\bibitem{karwowski2023sequential}
Jan Karwowski, Jacek Ma{\'n}dziuk, and Adam {\.Z}ychowski.
\newblock Sequential stackelberg games with bounded rationality.
\newblock {\em Applied Soft Computing}, 132:109846, 2023.

\bibitem{simulated_annealing}
Scott Kirkpatrick, C~Daniel Gelatt~Jr, and Mario~P Vecchi.
\newblock Optimization by simulated annealing.
\newblock {\em Science}, 220(4598):671--680, 1983.

\bibitem{blastshield}
Umesh Krishnaswamy, Rachee Singh, Nikolaj Bj{\o}rner, and Himanshu Raj.
\newblock Decentralized cloud wide-area network traffic engineering with
  {BLASTSHIELD}.
\newblock In {\em NSDI}, 2022.

\bibitem{bwe}
Alok Kumar et~al.
\newblock Bwe: Flexible, hierarchical bandwidth allocation for wan distributed
  computing.
\newblock In {\em SIGCOMM}, 2015.

\bibitem{pond}
Huaicheng Li, Daniel~S. Berger, Stanko Novakovic, Lisa Hsu, Dan Ernst, Pantea
  Zardoshti, Monish Shah, Samir Rajadnya, Scott Lee, Ishwar Agarwal, Mark~D.
  Hill, Marcus Fontoura, and Ricardo Bianchini.
\newblock Pond: Cxl-based memory pooling systems for cloud platforms, 2022.

\bibitem{sGames}
Tao Li and Suresh~P Sethi.
\newblock A review of dynamic stackelberg game models.
\newblock {\em Discrete \& Continuous Dynamical Systems-B}, 22(1):125, 2017.

\bibitem{raregan}
Zinan Lin, Hao Liang, Giulia Fanti, and Vyas Sekar.
\newblock Raregan: Generating samples for rare classes.
\newblock {\em arXiv preprint arXiv:2203.10674}, 2022.

\bibitem{self_driving}
Roland Meier, Thomas Holterbach, Stephan Keck, Matthias St\"{a}hli, Vincent
  Lenders, Ankit Singla, and Laurent Vanbever.
\newblock (self) driving under the influence: Intoxicating adversarial network
  inputs.
\newblock In {\em HotNets}. ACM, 2019.

\bibitem{ffd-1d-1991}
Yue Minyi.
\newblock A simple proof of the inequality $ffd(l) \leq 11/9opt(l)+1$, $\forall
  l$, for the ffd bin-packing algorithm.
\newblock {\em Acta Mathematicae Applicatae Sinca}, 1991.

\bibitem{synergy}
Jayashree Mohan, Amar Phanishayee, Janardhan Kulkarni, and Vijay Chidambaram.
\newblock Looking beyond {GPUs} for {DNN} scheduling on {Multi-Tenant}
  clusters.
\newblock In {\em OSDI}, 2022.

\bibitem{randalgo}
Rajeev Motwani and Prabhakar Raghavan.
\newblock {\em Randomized algorithms}.
\newblock Cambridge university press, 1995.

\bibitem{Z3}
Leonardo~de Moura and Nikolaj Bj{\o}rner.
\newblock {Z3}: An efficient {SMT} solver.
\newblock In {\em International conference on Tools and Algorithms for the
  Construction and Analysis of Systems}, pages 337--340. Springer, 2008.

\bibitem{metaopt}
Pooria Namyar, Behnaz Arzani, Ryan Beckett, Santiago Segarra, Himanshu Raj, and
  Srikanth Kandula.
\newblock Minding the gap between fast heuristics and their optimal
  counterparts.
\newblock In {\em HotNets}, 2022.

\bibitem{soroush}
Pooria Namyar, Behnaz Arzani, Srikanth Kandula, Santiago Segarra, Daniel
  Crankshaw, Umesh Krishnaswamy, Ramesh Govindan, and Himanshu Raj.
\newblock {S}olving {M}ax-{M}in {F}air {R}esource {A}llocations {Q}uickly on
  {L}arge {G}raphs.
\newblock In {\em NSDI}, 2024.

\bibitem{pop}
Deepak Narayanan, Fiodar Kazhamiaka, Firas Abuzaid, Peter Kraft, Akshay
  Agrawal, Srikanth Kandula, Stephen Boyd, and Matei Zaharia.
\newblock Solving large-scale granular resource allocation problems efficiently
  with {POP}.
\newblock In {\em SOSP}, 2021.

\bibitem{njw}
Andrew~Y Ng, Michael~I Jordan, and Yair Weiss.
\newblock {On spectral clustering: Analysis and an algorithm}.
\newblock In {\em {NIPS}}, 2002.

\bibitem{VBP-Heuristics}
Rina Panigrahy, Kunal Talwar, Lincoln Uyeda, and Udi Wieder.
\newblock Heuristics for vector bin packing.
\newblock January 2011.

\bibitem{cuckoo_break_without_knowing_hash_functions}
Pedro Reviriego and Daniel Ting.
\newblock Breaking cuckoo hash: Black box attacks.
\newblock {\em IEEE Transactions on Dependable and Secure Computing}, 2021.

\bibitem{sorting_net_1}
Thomas Sauerwald.
\newblock Sorting networks.
\newblock \url{https://www.cl.cam.ac.uk/teaching/1415/AdvAlgo/advalg.pdf}.

\bibitem{shoofly}
Rachee Singh, Nikolaj Bjorner, Sharon Shoham, Yawei Yin, John Arnold, and Jamie
  Gaudette.
\newblock Cost-effective capacity provisioning in wide area networks with
  {Shoofly}.
\newblock In {\em SIGCOMM}, 2021.

\bibitem{pifo}
Anirudh Sivaraman, Suvinay Subramanian, Mohammad Alizadeh, Sharad Chole,
  Shang-Tse Chuang, Anurag Agrawal, Hari Balakrishnan, Tom Edsall, Sachin
  Katti, and Nick McKeown.
\newblock Programmable packet scheduling at line rate.
\newblock In {\em Proceedings of the 2016 ACM SIGCOMM Conference}, SIGCOMM '16,
  page 44–57, New York, NY, USA, 2016. Association for Computing Machinery.

\bibitem{Abilene}
{Stanford University IT}.
\newblock Abilene core topology, 2015.

\bibitem{ffdsum-example}
Mark Stillwell, David Schanzenbach, Frédéric Vivien, and Henri Casanova.
\newblock Resource allocation algorithms for virtualized service hosting
  platforms.
\newblock {\em Journal of Parallel and Distributed Computing}, 2010.

\bibitem{ffddiv-ibm}
Chunqiang Tang, Malgorzata Steinder, Michael Spreitzer, and Giovanni Pacifici.
\newblock A scalable application placement controller for enterprise data
  centers.
\newblock In {\em WWW}, 2007.

\bibitem{jobpacking}
Abhishek Verma, Madhukar Korupolu, and John Wilkes.
\newblock Evaluating job packing in warehouse-scale computing.
\newblock In {\em 2014 IEEE International Conference on Cluster Computing
  (CLUSTER)}, 2014.

\bibitem{borg}
Abhishek Verma, Luis Pedrosa, Madhukar Korupolu, David Oppenheimer, Eric Tune,
  and John Wilkes.
\newblock Large-scale cluster management at google with borg.
\newblock In {\em EuroSys}, 2015.

\bibitem{wang2021fast}
Jiali Wang, He~Chen, Rujun Jiang, Xudong Li, and Zihao Li.
\newblock Fast algorithms for stackelberg prediction game with least squares
  loss.
\newblock In {\em International Conference on Machine Learning}, pages
  10708--10716. PMLR, 2021.

\bibitem{vbp-nphard}
Gerhard~J. Woeginger.
\newblock There is no asymptotic ptas for two-dimensional vector packing.
\newblock {\em Inf. Process. Lett.}, 1998.

\bibitem{ffdprod-sandpiper}
Timothy Wood, Prashant Shenoy, Arun Venkataramani, and Mazin Yousif.
\newblock Black-box and gray-box strategies for virtual machine migration.
\newblock In {\em NSDI}, 2007.

\bibitem{YenKLoopLess}
Jin~Y. Yen.
\newblock {F}inding the {K} {S}hortest {L}oopless {P}aths in a {N}etwork.
\newblock {\em Management Science}, 17(11):712--716, 1971.

\bibitem{aifo}
Zhuolong Yu, Chuheng Hu, Jingfeng Wu, Xiao Sun, Vladimir Braverman, Mosharaf
  Chowdhury, Zhenhua Liu, and Xin Jin.
\newblock Programmable packet scheduling with a single queue.
\newblock In {\em Proceedings of the 2021 ACM SIGCOMM 2021 Conference}, SIGCOMM
  '21, page 179–193, New York, NY, USA, 2021. Association for Computing
  Machinery.

\bibitem{zhang2021bayesian}
Yunxiao Zhang and Pasquale Malacaria.
\newblock Bayesian stackelberg games for cyber-security decision support.
\newblock {\em Decision Support Systems}, 148:113599, 2021.

\bibitem{arrow}
Zhizhen Zhong, Manya Ghobadi, Alaa Khaddaj, Jonathan Leach, Yiting Xia, and
  Ying Zhang.
\newblock Arrow: Restoration-aware traffic engineering.
\newblock In {\em SIGCOMM}, 2021.

\end{thebibliography}

\clearpage
\newpage
\appendix
\renewcommand\thefigure{A.\arabic{figure}}
\renewcommand\thesection{\Alph{section}}
\setcounter{subsection}{0}
\setcounter{section}{0}
\setcounter{figure}{0}
\setcounter{table}{0}
\renewcommand{\thetable}{A.\arabic{table}}
\renewcommand{\theequation}{\arabic{equation}}

\section{Details of Traffic Engineering}
\label{a:s:te}
\autoref{t:notation} summarizes our notation.

\begin{figure*}[t!]
	\begin{center}
		{\footnotesize
		\begin{tikzpicture}[->,>=stealth]
			\node[state, text width=6cm, label=DP Pseudocode] (A)
			{
				\begin{algorithmic}[t]
					\Input : $d_k$ requested rate of demand $k$ \EndInput
					\Input : $P_k$ paths for demand $k$ \EndInput
					\Input : $\hat{p}_k$ shortest path \EndInput
					\Input : $T_d$ demand pinning threshold \EndInput
					\ForAll{{$\text{demand}~k \in \mathcal{D}$}}
					\CIF{$d_k <= T_d$}
					\State \colorbox{myp}{$f^{\hat{p}_k}_k = d_k$}
					\EndIf
					\EndFor
					\State \colorbox{myb}{{\sf MaxFlow()}}
				\end{algorithmic}
			};
			
			\node[state,
			text width=7cm,
			right of=A,
			node distance=9.8cm,
			anchor=center, label={Modeling {\sf DP} in {\sf MetaOpt}}] (B)
			{%
				
				\begin{algorithmic}[t]
					\OuterVar : $d_k$ requested rate of demand $k$ \EndOuterVar
					\Input : $P_k$ paths for demand $k$ \EndInput
					\Input : $\hat{p}_k$ shortest path \EndInput
					\Input : $T_d$ demand pinning threshold \EndInput
					\ForAll{{$\text{demand}~k \in \mathcal{D}$}}
					\State \colorbox{myp}{${\sf ForceToZeroIfLeq}(d_k - f^{\hat{p}_k}_k, d_k, T_d)$}
					\EndFor
					\State \colorbox{myb}{{\sf MaxFlow()}}
				\end{algorithmic}
			};

			
			\begin{scope}[transform canvas={yshift=-0em}]
				\draw [->] (A) -- node[anchor=south,above] {${\sf MetaOpt}$} (B);
			\end{scope}
		\end{tikzpicture}
	}
	\end{center}
	\caption{The pseudo code for \dempin and how users can model it in \metaopt using the helper functions.\label{f:dp-helper-appendix}}
\end{figure*}

\begin{table}[t!]
	\centering
	{\footnotesize
		\renewcommand{\arraystretch}{1.5}
		\begin{tabular}{l p{2.15in}}
			{\bf Term} & {\bf Meaning}\\
			\hline
			\hline
			$\mathcal{V}, \mathcal{E}, \mathcal{D}, \mathcal{P}$ & Sets of nodes, edges, demands, and paths \\
			\hline
			$N, M, K$ & Number of nodes, edges, and demands, i.e., $N=|\mathcal{V}|, M=|\mathcal{E}|, K=|\mathcal{D}|$ \\\hline
			$c_e, p$ & {$c_e$: capacity of edge $e \in \mathcal{E}$\newline path $p$: set of connected edges}\\\hline
			$(s_k, t_k, d_k)$ & {The $k$th element in $\mathcal{D}$ has source and target nodes~($s_k, t_k \in \mathcal{V}$) and a non-negative volume~($d_k$)}\\\hline
			$\mathbf{f}, f_k^p$ & {$\mathbf{f}$: flow assignment vector with elements $f_k$ \newline $f_k^p$: flow for demand $k$ on path $p$} \\
			
		\end{tabular}
	}
	\caption{ Multi-commodity flow problems' notation.\label{t:notation}}
\end{table}

\subsection{Multi-commodity flow problem}
\label{sec::appendix-flow}

The optimal form of WAN-TE typically involves solving a multi-commodity flow problem. Given a set of nodes, capacitated edges, demands, and pre-chosen paths per demand, a flow allocation is feasible if it satisfies demand and capacity constraints. The goal is to find a feasible flow to optimize a given objective (e.g., total flow~\cite{ncflow}, max-min fairness~\cite{swan,b4}, or utility curves~\cite{bwe}). We define the feasible flow over a pre-configured set of paths as (see Table~\ref{t:notation})

{\small
	\begin{alignat}{4}
		&\multispan{4}\mbox{$\displaystyle \label{eq:feasible_flow} {\sf FeasibleFlow}(\mathcal{V}, \mathcal{E}, \mathcal{D}, \mathcal{P}) \triangleq \big\{\mathbf{f} \mid$}&\\
		&f_k = \sum_{p \in \mathcal{P}_k} f_k^p, & \forall k \in \mathcal{D} & \quad \text{(flow for demand $k$)}\nonumber\\
		&f_k \leq d_k, & \forall k \in \mathcal{D} & \quad \text{(flow below volume)}\nonumber\\
		&\sum_{k, p \mid p \in \mathcal{P}_k, e \in p} \!\!\!\!\! f_k^p \leq c_e, & \forall e \in \mathcal{E} & \quad \text{(flow below capacity)}\nonumber\\
		&f_k^p \geq 0 & \forall p \in \mathcal{P}, k \in \mathcal{D} & \quad \text{(non-negative flow)}\big\}\nonumber
\end{alignat}}
Among all the feasible flows, the optimal solution seeks to maximize the total flow across the network:

{\small
	\begin{align}
		\mbox{\sf OptMaxFlow}(\mathcal{V}, \mathcal{E}, \mathcal{D}, \mathcal{P}) \triangleq & \arg\max_{\mathbf{f}}\sum_{k \in \mathcal{D}} f_k\\
		\mbox{s.t.} \quad & \mathbf{f} \in {\sf FeasibleFlow}(\mathcal{V}, \mathcal{E}, \mathcal{D}, \mathcal{P}).\nonumber
\end{align}}


\subsection{More details on \dempin and \pop heuristics.}

\parab{Demand Pinning (\dempin)}~\cite{blastshield}. First, it routes all demands at or below a predefined threshold $T_d$ through their shortest path. It then jointly routes the rest of the demands optimally over multiple paths:

{\small
	\begin{alignat}{5}
		\multispan{4}\mbox{$\displaystyle \label{eq:demand_pinning} {\sf DemandPinning}(\mathcal{D}, \mathcal{P}) \triangleq \big\{\mathbf{f} \mid$} & \\
		& f_k^p & = & \begin{cases} d_k & \mbox{if }\ p\  \mbox{is shortest path in } P_k \\ 0 & \mbox{otherwise} \end{cases}, \,\, \forall k \in \mathcal{D}: d_k \leq T_d \big\}, \nonumber		
	\end{alignat}
	\vspace{-2mm}
}

We can pose {\dempin} as an optimization with constraints that route demands below the threshold on the shortest paths:

\begin{footnotesize}
	\vspace{-3mm}
	\begin{alignat}{3}
		\label{eq:dp_max_flow}
		\mbox{\sf DemPinMaxFlow}(\mathcal{V}, \mathcal{E}, \mathcal{D}, \mathcal{P}) \triangleq & \arg\max_{\mathbf{f}}\sum_{k \in \mathcal{D}} f_k & \\
		\mbox{s.t.} \quad & \mathbf{f} \in {\sf FeasibleFlow}(\mathcal{V}, \mathcal{E}, \mathcal{D}, \mathcal{P}) &\nonumber \\
		& \mathbf{f} \in {\sf DemandPinning}(\mathcal{D}, \mathcal{P}) &\nonumber
	\end{alignat}
	\vspace{-2mm}
\end{footnotesize}

\parab{Partitioned Optimization Problems (\pop).}~\cite{pop} \pop divides node pairs (and their demands) uniformly at random into partitions; assigns each partition an even share of edge capacities; and solves the original problem (e.g. the SWAN optimization~\cite{swan}) in parallel, once per partition. \pop results in a substantial speedup because it solves sub-problems that are smaller than the original (LP solver times are super-linear in problem sizes~\cite{boyd_co}) and it solves them in parallel.

{\small
	\begin{align}
		\label{eq:pop_max_flow}
		\mbox{\sf POPMaxFlow}&(\mathcal{V}, \mathcal{E}, \mathcal{D}, \mathcal{P}) \triangleq  \\
		& \bigcup_{\mbox{part. c}} \mbox{\sf OptMaxFlow}(\mathcal{V}, \mathcal{E}_c, \mathcal{D}_c, \mathcal{P}), \nonumber
	\end{align}
}

\noindent where $\cup$ is the vector union, the per-partition demands $\mathcal{D}_c$ are disjoint subsets of the actual demands drawn uniformly at random, and the per-partition edge list $\mathcal{E}_c$ matches the original edges but with proportionally smaller capacity.

%
%
%

\subsection{Formulation of \dempin and \pop}
\label{ss:convexify:te_1}
\label{sec:dempin_bigM}

\begin{figure}
	\includegraphics{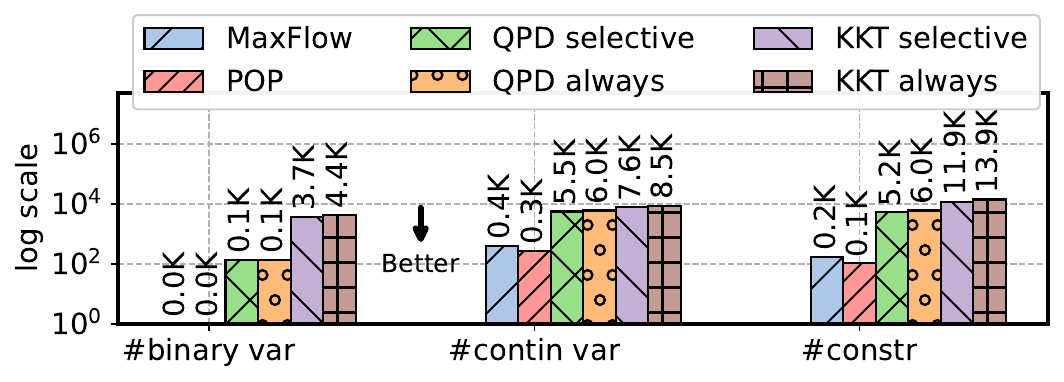}
	\caption{The complexity of user input of \pop and the subsequent rewrites in terms of the size of the optimization. \label{fig:complexity:pop}}
\end{figure}

\parab{Demand Pinning form for quantized demands.} \dempin has conditional or {\em if} clauses: if a demand is smaller than the threshold $T_d$, then route it over its shortest path; otherwise, use the optimal algorithm to route it. We reuse variables from the primal-dual rewrite to encode this:
{\small
	\begin{eqnarray}
		f_k^{\widehat{p}_k} \geq \sum_{q=1}^{Q} \indicator{L_{q} \leq T_d} L_{q}x^{k}_{q}, \,\,\, \forall k \in \mathcal{D}. \label{eqn:dempin_pdq}
	\end{eqnarray}
}
The $\{0, L_1, \dots,L_Q\}$ are the quantas and $x^k_q$'s are binary variables that pick which quanta is active for demand $k$ (for which $\widehat{p}_k$ is the shortest path).  The term $\indicator{L_{q} \leq T_d}$ is a no-op at runtime: it shows which terms exist in the sum. 

By definition of quantization: $d_k = \sum_q L_q x^k_q$. Thus, if the demand $d_k$ is no larger than the threshold $T_d$,~\eqnref{eqn:dempin_pdq} will ensure that its allocation on the shortest path equals the value of its demand.\footnote{When $d_k \leq T_d$,~\eqnref{eqn:dempin_pdq} effectively becomes $f_k^{\widehat{p_k}} \geq d_k$. In the converse case, \ie $d_k > T_d$, Eqn.~\ref{eqn:dempin_pdq} becomes $f_k^{\widehat{p}_k} \geq 0$, a no-op.}

We write this more directly and without quantized variables next. This is useful for the KKT rewrite where we don't do quantization or when users do not want to worry about rewrites but its big-$M$ based approach causes numerical instability.


\parab{Demand Pinning big-$M$ form.}
We can also encode \dempin using the standard big-$M$ approach from optimization theory~---~we need this for the KKT-encoding where we do not have quantized demands.

We describe the shortest path for demand $k$ using $\widehat{p}_k$ and write the {\em if} clause as:

{\small
	\begin{eqnarray}
		\sum_{p \in P_k,\ p \neq \widehat{p}_k} f_k^p \leq \max\left(M(d_k - T_d), 0\right), \,\,\, \forall k \in \mathcal{D}, \nonumber \\
		d_k - f_k^{\widehat{p}_k} \leq \max\left(M(d_k - T_d), 0\right), \,\,\, \forall k \in \mathcal{D}, \nonumber
	\end{eqnarray}
}
where $M$ is a large pre-specified constant.
Notice whenever the demand $d_k$ is below the threshold $T_d$ the constraints allocate zero flow on all but the (default) shortest path~---~\dempin routes the {\em full demand} on the default path in such cases. We can use standard optimization theory to convert the $\max$ in these constraints into a set of linear constraints~\cite{boyd_co} (the change also requires we modify the objective but does not impact the quality of the solution).

\parab{Partitioned Optimization Problems}. {\pop} is convex as it is the union of solutions to disjoint linear optimizations (\eqnref{eq:pop_max_flow}). It is hard to encode {\pop} as it uses random partitions, which makes ${\sf POP}(\mathcal{I})$ a random variable in the leader problem, but {\metaopt} needs a deterministic representation of the heuristic. 
Na\"ively, if we only consider a specific instance of each random variable (\eg for one instance of randomly assigning demands to partitions), it will overfit to that instance and not reflect {\pop}'s true behavior. 

We use the {\em expected value} or {\em tail percentile} of the gap from multiple random trials. To compute the average, we replace \heur in \eqnref{eq:metaopt_full_1} with its expected value and approximate the expectation through empirical averages over a few randomly generated partitions (see \secref{s:eval:te}). To find the tail, we use a sorting network~\cite{sorting_net_1,pretium} to compute the desired percentile across multiple random trials. 

In addition, we encode an advanced version of \pop in~\xref{s:appendix:pop} that splits large demands across multiple partitions instead of assigning each demand to one partition.

\subsection{{\pop} Client Splitting}
\label{s:appendix:pop}

In~\secref{s:background}, we introduce the (basic) \pop heuristic~\cite{pop}, which incorporates \emph{resource splitting} for our WAN TE problem, and in~\secref{ss:convexify:te_1}, we present \pop as a convex optimization.
The work in~\cite{pop} also specifies an extended full-fleshed version of \pop that incorporates ``\emph{client splitting}''. We next show how to express this variant as a convex optimization problem.

We can think of \pop client splitting as an operation that takes in a set of demands $\mathcal{D}$ and returns a modified set $\mathcal{D}_{\mathrm{cs}} = \mathrm{ClientSplit}(\mathcal{D})$ that can then be input into \pop as in~\eqref{eq:pop_max_flow}.
The function $\mathrm{ClientSplit}()$ generates several duplicates of the existing demands and reduces their volume in proportion. It performs several operations where it replaces $(s_k, t_k, d_k) \in \mathcal{D}$ with two elements of the form $(s_k, t_k, d_k/2)$. It iterates and repeats this operation until it terminates (see~\cite{pop}).


We encode a version of client splitting where we split an element in $\mathcal{D}$ if its demand value $d_k$ is larger than or equal to a threshold $d_{\mathrm{th}}$, and we keep splitting it until either we get to a predefined number of maximum splits (of the original demand\footnote{Notice that~\cite{pop} pre-specifies a total aggregated number of splits across all clients whereas we set the a maximum for per-client splits. This slight modification facilitates the convex representation of the heuristic.}) or the split demand is lower than $d_{\mathrm{th}}$.

Without loss of generality, we describe this idea for a single demand $d_1$: we can replicate this process for all demands in $\mathcal{D}$. Take for example the scenario where we split the demand at most twice (which creates at most $4$ virtual clients). We a-priori encode all the flow variables for all possible splits of $d_1$: we use seven variables of the form $f_{i,j}$ instead of the single $f_i$ where $f_{1_1}$ is the flow if we do not split the client; $f_{1,2}$ and $f_{1,3}$ are the flows if we split the client once and $f_{1,4}$ through $f_{1,7}$ are the flows if we split the client twice. These flows have to satisfy:
\begin{align*}
	0 & \leq f_{1,1} \leq d_1, \\
	0 & \leq f_{1,i} \leq \frac{d_1}{2}, \,\,\, \text{for } i \in \{2, 3\} \\
	0 &\leq f_{1,i} \leq \frac{d_1}{4},  \,\,\, \text{for } i \in \{4, 5, 6, 7\}. 
\end{align*}
The variable $f_{1,1}$ should be zero unless $d_1 < d_{\mathrm{th}}$ (when we do not split clients), which we can achieve using big-$M$ constraints (~\secref{ss:convexify:te_1}):
\begin{align*}
	f_{1,1} \leq \max( M (d_{\mathrm{th}}-d_1), 0).
\end{align*}
We want $f_{1,2}$ and $f_{1,3}$ to be exactly zero unless $d_1 \geq d_{\mathrm{th}}$ and $d_1/2 < d_{\mathrm{th}}$, which we can achieve by doing
\begin{align*}
	f_{1,i} & \leq \max( M (d_1-d_{\mathrm{th}}+\epsilon), 0),  && \,\,\, \text{for } i \in \{2, 3\}, \\
	f_{1,i} & \leq \max( M (d_{\mathrm{th}}-d_1/2), 0),  && \,\,\, \text{for } i \in \{2, 3\},
\end{align*}
where we added the small pre-specified $\epsilon >0$ to allow for the case where $d_1 = d_{\mathrm{th}}$.
Lastly, we want $f_{1,4}$ through $f_{1,7}$ to be exactly zero unless $d_1 \geq 2 d_{\mathrm{th}}$.
We encode this as:
\begin{align*}
	f_{1,i} & \leq \max( M (d_1-2d_{\mathrm{th}}+\epsilon), 0),  && \,\,\, \text{for } i \in \{4, 5, 6, 7\}.
\end{align*}
We can replicate this procedure for all $d_k$ and encode \pop with client splitting as a convex optimization problem. Once this is done, the techniques in~\secref{s:d:rewrites} apply.

\section{Details of Vector Bin Packing}
\label{a:s:vbp}

\subsection{Formulation of FFD (First Fit Decreasing)}
\label{a:s:vbp_new}
\label{ss:convexify:vbp_1}

\newcommand{\fit}{f}

We formulate the first-fit-decreasing heuristic as a feasiblity problem over a set of constraints.
That is, there will be no objective term. Such a formulation allows the two-level optimization problem \metaopt solves to become a one-shot optimization without a rewrite (because the heuristic follower has no objective, there can be no mismatch with the objective of the meta-optimization leader). 
To our knowledge, this formulation of FFD is novel. 

\parab{Modeling \ffd using \metaopt's helper functions.} \autoref{f:ffd-helper-appendix}, right show how users can easily model \ffd without having to go through the mathematical details. It also show the mapping between the helper functions and the pseudo code in \ffd using different colors.

\parab{Details of how we model \ffd as an optimization.}
\autoref{t:ffd_notation} lists our notation. 
The model uses integer (or binary) variables, is not a scalable method to solve FFD in practical systems, and we propose it only as an effective method to find adversarial inputs for the FFD heuristic.

\begin{table}[t!]
\begin{tabular}{ll}
{\bf Term} & {\bf Meaning}\\\hline
$i, j, d$ & {indexes for ball, bin and dimension}\\
${\bf Y}_i, {\bf C}_{j}$ & {Multi-dim vectors of ball and bin sizes}\\
$W_i$ & {Weight of ball $i$}\\
$\alpha_{ij}$ & {$=1$ if ball $i$ is assigned to bin $j$ and $0$ otherwise}\\
${\bf x}_{ij}$ & {Vector of resources allocated to ball $i$ in bin $j$}\\
${\fit}_{ij}$ & {$=1$ if ball $i$ can fit in bin $j$ and $0$ otherwise}\\\hline
\end{tabular}
\caption{\label{t:ffd_notation} The notation we use to formulate FFD as a feasibility problem. We use bold to indicate multi-dimensional vectors and capitalize variables which are typically constants for FFD but in the case of this work can be variables of the outer meta-optimization problem.}
\end{table}

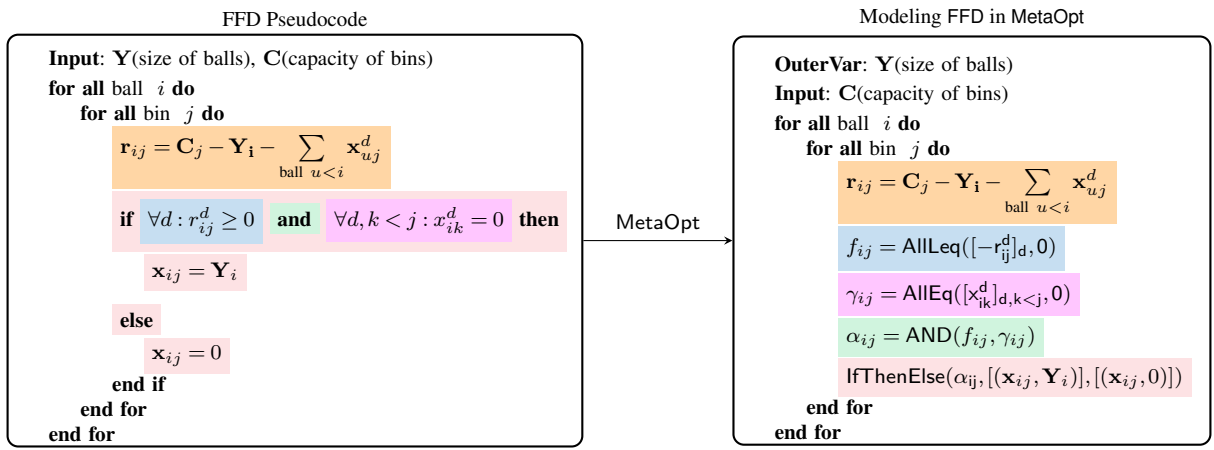
\begin{figure*}[t!]
	\begin{center}
		{\footnotesize
		\begin{tikzpicture}[->,>=stealth]
			\node[state, text width=7.5cm, label=FFD Pseudocode] (A)
			{
				\begin{algorithmic}[t]
					\Input : ${\bf Y}$(size of balls), ${\bf C}$(capacity of bins) \EndInput
					\ForAll{$\text{ball}~~i$}
					\ForAll{$\text{bin}~~j$}
					\State \colorbox{myo}{${\bf r}_{ij} = {\bf C}_j - {\bf Y_{i}} - \sum\limits_{\text{ball}~u < i}{{\bf x}_{uj}^d}$}
					\CIF{\colorbox{myb}{$\forall d:  r^d_{ij} \geq 0$} \colorbox{mygreen}{{\bf and}} \colorbox{mymagenta}{$\forall d, k < j: x^d_{ik} = 0$}}
					\State \colorbox{myp}{${\bf x}_{ij} = {\bf Y}_i$} \\
					\CELSE
					\State \colorbox{myp}{${\bf x}_{ij} = 0$}
					\EndIf
					
					\EndFor
					\EndFor
				\end{algorithmic}
			};
			
			\node[state,
			text width=6.2cm,
			right of=A,
			node distance=9.0cm,
			anchor=center, label={Modeling {\sf FFD} in {\sf MetaOpt}}] (B)
			{%
				\begin{algorithmic}[t]
					\OuterVar  : ${\bf Y}$(size of balls)  \EndOuterVar
					\Input : ${\bf C}$(capacity of bins) \EndInput
					\ForAll{$\text{ball}~~i$}
					\ForAll{$\text{bin}~~j$}
					\State \colorbox{myo}{${\bf r}_{ij} = {\bf C}_j - {\bf Y_{i}} - \sum\limits_{\text{ball}~u < i}{{\bf x}_{uj}^d}$}
					\State \colorbox{myb}{$f_{ij} = {\sf AllLeq([-r^d_{ij}]_d, 0)}$}
					\State \colorbox{mymagenta}{$\gamma_{ij} = {\sf AllEq([x^d_{ik}]_{d, k<j}, 0)}$}
					\State \colorbox{mygreen}{$\alpha_{ij} = {\sf AND}(f_{ij}, \gamma_{ij})$}
					\State \colorbox{myp}{${\sf IfThenElse}({\sf \alpha_{ij}}, [({\bf x}_{ij}, {\bf Y}_i)], [({\bf x}_{ij}, 0)])$}
					\EndFor
					\EndFor
				\end{algorithmic}
			};

			
			\begin{scope}[transform canvas={yshift=-0em}]
				\draw [->] (A) -- node[anchor=south,above] {${\sf MetaOpt}$} (B);
			\end{scope}
		\end{tikzpicture}
	}
	\end{center}
	\caption{The pseudo code for \ffd and how users can model it in \metaopt using the helper functions.\label{f:ffd-helper-appendix}}
\end{figure*}

\vspace{0.05in}
\noindent{\bf Modeling {\em decreasing} ball weights: }
Recall that in each iteration FFD assigns the ball with the largest weight among the unassigned balls.
We can model each iteration and use a sorting network~\cite{sorting_net_1} to ensure we assigns balls in decreasing order in each iteration. We propose a simpler alternative:

We observe the ball-weighting functions are a fixed function of the ball size: with {\sf{FFDSum}}~\cite{ffdsum-example}, {\sf{FFDProd}}~\cite{ffdprod-sandpiper}, {\sf{FFDDiv}}~\cite{ffddiv-ibm}, if ${\bf Y}$ is a multi-dimension vector and captures the size of the $i$'th ball on each dimension, then the weight of the $i$'th ball, $W_i$, is $\sum_d Y_i^d$, $\prod_d Y_i^d$ and $Y_i^1/ Y_i^2$ respectively.\footnote{The {\sf FFDDiv} function applies only on two dimensions.}

We constrain the input space (${\sf ConstrainedSet}$ in \eqnref{eq:metaopt_full_1}) to ensure that if we assign balls to bins based on their index we also assign them in decreasing order of weight:

{\small
	\begin{equation}
	W_{i} \geq W_{i + 1} \,\,\, \forall{\sf item}\,\, i \label{eq:ffd_weights}
	\end{equation}
}

\begin{table}[t!]
\begin{tabular}{l|cccccc}
& \multicolumn{6}{l}{\bf Bin index $j\longrightarrow$}\\\hline
Fit ${\fit}_{ij}$ 					& $0$ & $0$ 		& $1$ & $0$ 		& $1$ 		& $1$\\
RHS of Eqn.~\ref{eq:ffd_alpha} 	& $0$ & $\frac{1}{2}$ 	& $1$ & $\frac{2}{4}$	& $\frac{4}{5}$	& $\frac{4}{6}$\\
$\alpha_{ij}$ 				& $0$ & $0$			& $1$ & $0$			& $0$			& $0$\\
\end{tabular}
\caption{Illustrating how we model {\em first-fit} using Eqn.~\ref{eq:ffd_alpha} \label{t:alpha_constraint}}
\end{table}

\vspace{0.05in}
\noindent{\bf Modeling {\em first-fit}:}  FFD assigns each ball to the first bin that it can fit in. 
Let bins be ordered in index order and $\alpha_{ij}$ be an integer variable whose value is $1$ iff the first bin (i.e., the one with the smallest index) that ball $i$ fits in is bin $j$.
We model the first-fit constraint as:
{\small 
	\begin{alignat}{3}
	\alpha_{ij} & \leq \frac{{\fit}_{ij} + \sum_{\mbox{bin}\ k < j}(1- {\fit}_{ik})}{j} & ~~~\forall{\sf item} i,\ \forall{\sf bin}\ j \label{eq:ffd_alpha}\\
	\sum\limits_{\mbox{bin} \ j}\alpha_{ij} & = 1 & ~~~\forall{\sf item}\,{i}
	\end{alignat}
}
\autoref{t:alpha_constraint} shows an example that illustrates this constraint in action. 
It is easy to prove that the right-hand-side of Eqn.~\ref{eq:ffd_alpha} is $1$ for the first bin where the ball can fit (i.e., smallest index in set of bins $\{j \|\ \mbox{fit}\ {\fit}_{ij} = 1\}$) and less than $1$ for all other bins. The second constraint is necessary to ensure that $\alpha_{ij} = 1$ for exactly the first-fitting bin for each ball.

\vspace{0.05in}
\noindent{\bf Modeling resource allocation and capacity constraints:} We first make sure we allocate resources in a way that is consistent with ball assignment: we should allocate sufficient resources to a ball from the bin we assigned it to and ensure we do not assign it resources from any other bin. We can do this simply by: ${\bf x}_{ij} \triangleq {\bf Y}_i \alpha_{ij}$. Here, the resource assigned to a ball $i$ at bin $j$, ${\bf x}_{ij}$ is simply the product of the ball size vector ${\bf Y}$ with the assignment indicator variable $\alpha$. But ${\bf Y}$ is also a variable in the outer (or leader) problem and so such equations are non-linear. To linearize these equations, we use a technique similar to what is colloquially known as big-M in optimization literature. Let $Z$ be an appropriately large postive integer constant, then:
{\small
	\begin{alignat}{3}
	x^{d}_{ij} & \leq Z \alpha_{ij} & ~~~\forall {\sf item}\, i \,\,\forall {\sf bin}\, j \,\,\forall {\sf dim}\, {d}\\
	\sum_{j} x_{ij}^d & = Y^{d}_{i} & ~~~ \forall{\sf item}\, i \,\,\forall {\sf dim}\, {d}
	\end{alignat}
}

We define the residual capacity of bin $j$ when we allocate ball $i$ to it as (remember we allocate ball $i$ only after we allocate all balls with a lower index):
{\small
	\begin{equation}
		r_{ij} \triangleq {\bf C}_j - {\bf Y_i} - \sum_{\mbox{balls}\ u < i}{\bf x}_{uj}~~~\forall {\sf item}\, i \,\,\forall {\sf bin}\, j \label{eq:ffd_residual}
	\end{equation}
}
The sum on the right captures how much resources we have already allocated to other balls from this bin. We define fit variables ${\fit}$ to enforce capacity constraints. We next ensure $\fit_{ij}$ is $1$ iff the bin $j$ has adequate resources to fit ball $i$. To do so, let $M$ be some appropriately large constant positive integer, we use: 
{\small
	\begin{equation}
	\min_d r_{ij}^d \leq M \fit_{ij} \leq M + \min_d r_{ij}^d, ~~~\forall {\sf item}\, i \,\,\forall {\sf bin}\, j\,\,\forall {\sf dim}\, {d} \label{eq:ffd_fit}
	\end{equation}
}
Here, if the ball $i$ fits in bin $j$ then the residual capacity $r_{ij}^d$ is greater than $0$ across all dimensions $d$ and Eqn.~\ref{eq:ffd_fit} clamps $M{\fit}$ between a positive number and $M$ plus that positive number (remember $r_{ij}$ is the capacity of bin $j$ that remains {\em after} we assign ball $i$ to it).\footnote{Corner case: when the residual capacity is precisely $0$ on all dimensions, we want $\fit$ to still be $1$ but these constraints will allow $\fit$ to be $0$. This is a rare case but it can occur in practice; we can solve this corner case in a few different ways including for example adding a small value $\epsilon$ to the left-most term in equation~\ref{eq:ffd_fit}.} 
Since ${\fit}$ is a binary variable, the only feasible assignment in this case is $1$.
Conversely, if the ball does not fit in a bin, the residual capacity $r_{ij}^d$ is below $0$ on at least one dimension $d$ and the constraint in Eqn.~\ref{eq:ffd_fit} clamps $M{\fit}$ to be between a negative number and $M$ plus that negative number which forces $\fit$ to be $0$.
In practice we set the value of $M$ to be larger than the largest single-dimension bin capacity~($i.e., \max_{j, d} C_j^d$).

\vspace{0.05in}
\noindent{\bf Unique solution for FFD: }The equations~\ref{eq:ffd_alpha}--\ref{eq:ffd_fit} uniquely specify a solution to the iterative first-fit-decreasing heuristic. These constraints are linear even if the ball and bin sizes are variables in the outer problem. This is key because {\metaopt} can apply without having to rewrite the heuristic follower.

\vspace{0.05in}
\noindent{\bf Counting the number of bins used:}
To find an adversarial gap, the outer (leader) problem needs the number of bins used by {\ffd} since, in this case, {\metaopt} seeks inputs which cause the heuristic to use more bins than the optimal: 
{\small
\begin{equation}
\mbox{Num. bins used by {\ffd}} \triangleq \sum_{\mbox{bin}\ j} \max\limits_{\mbox{ball}\ i} \alpha_{ij}.
\end{equation}
} 
The term simply counts bins to which at least one ball is assigned.
Clearly, this term is linear ($\max$ has a linear rewrite) and does not give rise to any additional concerns.

\subsection{Proof of \theoremref{th:ffdsum:approx-ratio}}
\label{apendix:proof:ffdsum}

Our goal is to show that for any $k > 1$, an input $\mathcal{I}$ exists where ${\sf FFDSum}(\mathcal{I})$ needs at least $2k$ bins while ${\sf OPT}(\mathcal{I})$ only needs $k$. 
Since we are proving a lower bound on the approximation ratio of ${\sf FFDSum}$, it suffices to show an example for each $k$. 
We do this by the following; for every value of $k > 1$, we can find $m$ and $p$ such that $k = 2m + 3p$ and $p \in \{0, 1\}$. 
Then, we create an example consisting of $6m + 9p$ balls where ${\sf FFDSum}(\mathcal{I}) = 2{\sf OPT}(\mathcal{I}) = 2k$. We show the constructed example in \tabref{tab:ffd-2d:proof-example} along with the allocation from ${\sf OPT}$ and ${\sf FFDSum}$.

\begin{table}[t!]
	\centering
	\begin{tabular}{|c|c|c|c|c|c|}
		\hline
		Ball & Ball  	& 	Wei 	& 	Num	&	Bin ID  	& Bin ID \\ 
		ID		&	Size	&	ght		&		&	({\sf OPT})		& (\ffd) \\	
		\hline \hline
		1	&	[0.92, 0.00]	&	0.92	&	$m$	& 	B1	&	B1	\\ \hline
		2	&	[0.91, 0.01]	&	0.92	&	$m$	&	B2	&	B2	\\ \hline
		3	&	[0.48, 0.20]	&	0.68	&	\rdelim\}{4}{*}[${}\times p$]	&	C1	&	C1	\\ \cline{1-3}\cline{5-6}
		4	& 	[0.68, 0.00]	&	0.68	&		&	C2	&	C2	\\ \cline{1-3}\cline{5-6}
		5	& 	[0.52, 0.12]	&	0.64	&		&	C3	&	C1	\\ \cline{1-3}\cline{5-6}
		6	& 	[0.32, 0.32]	&	0.64	&		&	C3	&	C2	\\ \hline
		7	& 	[0.19, 0.45]	&	0.64	&	\rdelim\}{2}{*}[${}\times p$]	&	C2	&	C3	\\ \cline{1-3}\cline{5-6}
		8	& 	[0.42, 0.22]	&	0.64	& 		&	C1	&	C3	\\ \hline
		9	& 	[0.10, 0.54]	&	0.64	&	$p$	&	C1	&	C4	\\ \hline
		10	& 	[0.10, 0.54]	&	0.64	&	$p$	&	C2	&	C5	\\ \hline
		11	& 	[0.10, 0.53]	&	0.63	&	$p$	&	C3	&	C6	\\ \hline
		12	& 	[0.06, 0.48]	&	0.54	&	$m$	&	B2	&	B1	\\ \hline
		13	& 	[0.07, 0.47]	&	0.54	&	$m$	&	B1	&	B2	\\ \hline
		14	& 	[0.01, 0.53]	&	0.54	&	$m$	&	B1	&	B3	\\ \hline
		15	& 	[0.03, 0.51]	&	0.54	&	$m$	&	B2	&	B4	\\ \hline
	\end{tabular}
	\caption{Constructed example to prove the approximation ratio of 2d-${\sf FFDSum}$ is always lower bounded by 2 for any value of ${\sf OPT}(\mathcal{I})>1$. 
		The notation $^A_B\}\times p$ represents a sequence of $2p$ balls  consisting of ball type $A$ and $B$ where one occurrence of $A$ is always followed by one  of $B$. (\eg for $p=2$, this sequence is equivalent to $ABAB$). \label{tab:ffd-2d:proof-example}}
\end{table}
\section{Details of Packet Scheduling Heuristics}
\label{s:a:packet-scheduling}
Here, we describe how we model SP-PIFO~\cite{sp-pifo} as a feasibility problem. \aifo is similar. These formulations are to the best of our knowledge novel. \autoref{t:ps_notation} lists our general notations and \autoref{t:sp_pifo_notation} and \autoref{t:aifo_notation} show our specific notations for SP-PIFO and AIFO respectively.

\parab{Definition (Ranks and Priorities)}. Packet scheduling papers~\cite{sp-pifo,aifo} use both ranks and priorities: a packet with a higher rank has a lower priority and vice-versa. If a packet has rank $R_p$, and $R_{max}$ is the maximum possible rank, we can compute the priority of the packet by $R_{max} - R_p$ which ensure all the packets with rank lower than $R_i$ have higher priority value and all the packets with rank higher than $R_i$ have lower priority value.

\begin{table}[t!]
	\begin{tabular}{ll}
		{\bf Term} & {\bf Meaning}\\\hline
		$P, p$ & {number of packets and index for packet}\\
		$R_{max}, R_p$ & {maximum rank and rank of packet $p$}\\
		$w_p$ & {Weight of packet $p$}\\
		$d_{pj}$ & {$1$ if packet $p$ dequeued after $j$, o.w. $=0$}\\\hline
	\end{tabular}
	\caption{\label{t:ps_notation} Notation used to formulating packet scheduling heuristics as a feasibility problem. We capitalize variables which are typically constants for heuristic but in the case of this work can be variables of the outer meta-optimization problem.}
\end{table}

\subsection{Formulation of SP-PIFO}
\label{a:s:sp-pifo}

\begin{figure*}[t!]
	\begin{center}
		{
			\footnotesize
		\begin{tikzpicture}[->,>=stealth]
			\node[state, text width=6cm, label=SP-PIFO Pseudocode] (A)
			{
				\begin{algorithmic}[t]
					\Input : incoming packets ${\bf p}$ with rank $R_p$ \EndInput
					\Input : number of queues $N$ \EndInput
						\State ${\bf l}^0 = {\bf 0}$
						\ForAll{$\text{packet}~p$}
						\CIF{\colorbox{myb}{$R_p < l^{p-1}_N$}}
						\State \colorbox{myp}{${\bf \hat{l}}^p = {\bf l}^{p-1} + (R_p - l^{p-1}_N)$}
						\EndIf
						\For{$\text{queue}~q \in \{1,\dots,N\}$}
						\CIFG{\colorbox{mymagenta}{$\hat{l}^p_q \leq R_p < \hat{l}^p_{q-1}$}}
						\State \colorbox{mygreen}{$x^p_{q} = 1$}
						\State \colorbox{mygreen}{$l^{p}_q = R_p$}
						\CELSEG
						\State \colorbox{mygreen}{ $x^p_{q} = 0$}
						\EndIf
						\EndFor
						\EndFor
				\end{algorithmic}
			};
			
			\node[state,
			text width=7.5cm,
			right of=A,
			node distance=9.0cm,
			anchor=center, label={Modeling {\sf SP-PIFO} in {\sf MetaOpt}}] (B)
			{%
				\begin{algorithmic}[t]
						\OuterVar  : incoming packets ${\bf p}$ with rank $R_p$ \EndOuterVar
						\Input number of queues $N$ \EndInput
						\State ${\bf l}^0 = {\bf 0}$
						\ForAll{$\text{packets}~p$}
						\State \colorbox{myb}{$\alpha_p = {\sf IsLeq}(R_p - l_N^{p-1} - \epsilon, 0)$}
						\State \colorbox{myp}{${\sf IfThen}(\alpha_p, [({\bf \hat{l}}^p, {\bf l}^{p - 1} + R_p - l_N^{p-1} )])$}
						\ForAll{$\text{queue}~q \in \{1,\dots,N\}$}
						\State \colorbox{mymagenta}{$u_{pq} = {\sf IsLeq}(R_p - \hat{l}^p_{q-1} - \epsilon, 0)$}
						\State \colorbox{mymagenta}{$l_{pq} = {\sf IsLeq}(\hat{l}^p_{q} - R_p, 0)$}
						\State \colorbox{mymagenta}{$f_{pq} = {\sf AND}(u_p, l_p)$}
						\State \colorbox{mygreen}{${\sf IfThenElse}(f_{pq}, [(x^p_q, 1), (l_q^p, R_p)], [(x^p_q, 0)])$}
						\EndFor
						\EndFor
				\end{algorithmic}
			};

			
			\begin{scope}[transform canvas={yshift=-0em}]
				\draw [->] (A) -- node[anchor=south,above] {${\sf MetaOpt}$} (B);
			\end{scope}
		\end{tikzpicture}
	}
	\end{center}
	\caption{The pseudo code for SP-PIFO and how users can model it in \metaopt using the helper functions.\label{f:sp-pifo-helper}}
\end{figure*}
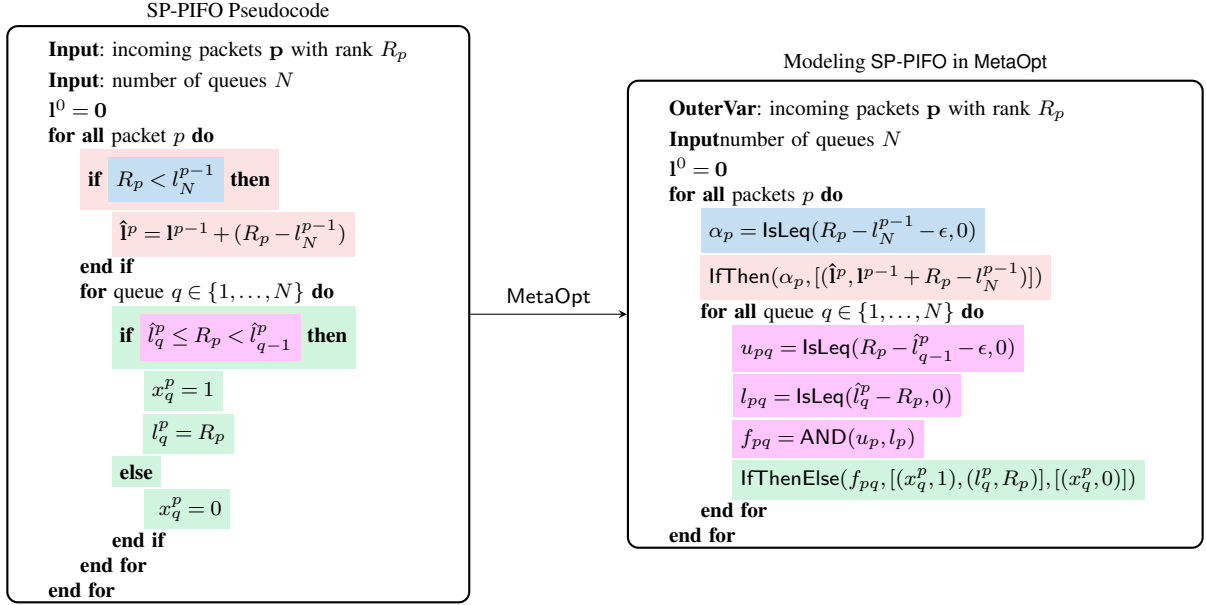

SP-PIFO approximates PIFO~\cite{pifo} using $n$ strict priority FIFO queues. It keeps a packet rank for each queue (\ie queue rank) that shows the lower bound on packet ranks the queue admits. For each packet, it starts from the lowest priority queue until it finds the first queue that can admit the packet (packet rank $\geq$ queue rank). If a queue admits the packet, SP-PIFO adds the packet to the queue and updates the queue rank to the recently admitted packet's rank (\ie push up). If none of the queues admit the packet (packet rank $<$ highest-priority queue rank), it reduces the rank of all the queues such that the highest-priority queue can admit the packet (\ie push down). \autoref{f:sp-pifo-helper}, left shows the pseudo code for SP-PIFO.

\parab{Modeling SP-PIFO using \metaopt's helper functions.} \autoref{f:sp-pifo-helper}, right show how users can easily model SP-PIFO without having to go through the mathematical details. It also show the mapping between the helper functions and the pseudo code in SP-PIFO using different colors.

\parab{Modeling push down.}  SP-PIFO reduces the rank of all the queues if none of the queues can admit the packet. This happens when the rank of the highest priority queue is higher than the packet rank ($R_p$). We model this as:
{\small
	\begin{equation}
		\hat{l}^p_q = l^{p-1}_q + max(0, l^{p-1}_N - R_p)
	\end{equation}
}

This constraint keeps the queue ranks the same if the packet rank $R_p$ is greater than the highest priority queue $l^{p-1}_N$. Otherwise, it applies push down and reduces the rank of all the queues so that the highest priority queue can admit the packet (after the update, the rank of highest priority queue $N$ is the same as packet rank $R_p$).
%
%
%

\parab{Deciding on the proper queue.} Recall SP-PIFO adds a packet to the queue with the lowest priority among the ones that can admit the packet; that is the $q$ such that $R_p \geq \hat{l}^p_q$ (admits the bin) and the one lower priority queue $q-1$ does not admit the packet ($R_p < \hat{l}^p_{q - 1}$). We model this as following:
{\small
	\begin{alignat}{3}
		Mx_{pq} &\leq M + R_p - \hat{l}^p_q & ~~~\forall{\sf packet}\, p\,\, \forall{\sf queue}\, q\label{eq:which-queue-1} \\
		Mx_{pq} &\leq M + \hat{l}^p_{q - 1} - R_p - \epsilon & ~~~\forall{\sf packet}\, p\,\, \forall{\sf queue}\, q\label{eq:which-queue-2} \\
		\sum_{q} x_{pq} & = 1 & ~~~\forall{\sf packet}\,p\label{eq:which-queue-3} 
	\end{alignat}
}
where $M$ is a large constant ($\geq R_{max}$) and $\epsilon$ is a small constant ($<1$). The first constraint ensures a queue with rank greater than the rank of the packet does not get the packet (if $R_p < \hat{l}^p_q$, the constraint forces $x_{pq}$ to 0). The second constraint ensures a queue does not get the packet if a lower priority queue admits the packet (if $R_p \geq \hat{l}^p_{q - 1}$, the constraint forces $x_{pq}$ to). The last constraint forces the optimization to place the packet in one of the queues.

\begin{table}[t!]
	\begin{tabular}{ll}
		{\bf Term} & {\bf Meaning}\\\hline
		$N, q$ & {number of queues, and index for queue}\\
		$B$ & {burst factor in \aifo} \\
		${\bf l}^{p-1}$ & {vector of queue ranks when deciding for packet $p$}\\
		${\bf \hat{l}}^{p}$ & {vector of queue ranks after push down for packet $p$} \\
		$x_{pq}$ & {=1 if packet p in queue q, o.w. =0} \\\hline
	\end{tabular}
	\caption{\label{t:sp_pifo_notation} Additional Notations for SP-PIFO.}
\end{table}

\parab{Modeling push up.} Recall SP-PIFO updates the rank of the queue to the most recently admitted packet's rank. We model this as:
{\small
\begin{equation}
	l^{p}_q = \hat{l}^p_q + x_{pq} (R_p - \hat{l}^p_q)\,\,\,\forall{\sf packet}\,p\,\,\,\forall{\sf queue}\,q \label{eq:push-up}
\end{equation}
}

This constraint only updates the rank of queue $q$ to $R_p$ if the packet is place in the queue ($x_{pq} = 1$). We can linearize this constraint~\cite{boyd_co}.

\parab{Unique solution for SP-PIFO.} We can combine these constraints to uniquely specify SP-PIFO's decisions on a sequence of incoming packets. All these constraints are linear or linearizable using standard techniques even though packet ranks are variables in the outer problem. 

\parab{Computing weighted average delay.} We measure the gap in terms of the average delay of forwarding packets weighted by their priority (inverse of their rank). To measure delay of a packet, we count how many packets SP-PIFO decides to dequeue before it. Let $d_{pj}$ indicate whether packet $p$ is dequeued after packet $j$, we model the weighted average delay as:
{\small
\begin{equation}
	\text{Weighted avg delay} = \frac{1}{N} \sum_{{\sf pkt}\, p, j \neq p} (R_{max} - R_p) {d_{pj}} \label{eq:sp-pifo-cost}
\end{equation}
}

Next, we define $d_{pj}$: we first assign weights to the packets such that the weights respect the order in which the packets should be dequeued (a packet $j$ has a higher weight than $p$ if it should be dequeued before $j$). We assign weights $w_p$ as:

{\small
\begin{equation}
	w_p = -p + \sum_{{\sf queue}\, q}{qPx_{pq}}
\end{equation}
}

This weighting guarantees that (1) a packet from a higher priority queue always has a higher weight than a packet from a lower priority queue, and (2) among the packets in the same priority queue, the one arrived earlier has a higher weight. Packet $p$ is dequeued after packet $j$ if the weight of packet $j$ is higher.
{\small
	\begin{equation}
	 	w_j - w_p \leq Md_{pj} \leq M + w_j - w_p\,\,\,\forall{\sf packets}\,p,j
	\end{equation}
}
Note that weights are unique.

\subsection{Formulation of AIFO}
\label{a:s:aifo}
\aifo~\cite{aifo} is an admission control on top of a FIFO queue that tries to approximate the same set of packets a \pifo queue would admit and is specifically designed for shallow buffers. AIFO keeps a window of recently seen packet ranks and computes the relative rank of the new packet with respect to this window. Then, it compares this quantile estimate with the fraction of available space in the queue (multiplied by some constant burst factor). If the quantile is lower or equal, \aifo admits the packet. Otherwise, it drops the packet.

\parab{Finding quantile estimate.} Recall that \aifo computes how many packets in its recent window have lower ranks than an incomming packet $p$, we model this as:

{\small
\begin{alignat}{2}
	R_p - R_j \leq M g_{pj} \leq M + R_p - R_j - \epsilon & \label{eq:quantile} \\ 
	& \hspace{-30mm}\forall{\sf packet}\,p\,\forall{\sf packet}\, j:\,\,\, p - K \leq j \leq p - 1\nonumber \\
	g_{p} = \sum_{{\sf pkt}~j: p - K \leq j \leq p}{g_{pj}} \label{eq:total_quantile}
\end{alignat}
}

\noindent where $M$ is a large constant ($M \geq R_{max}$) and $\epsilon$ is a small constant ($\epsilon < 1$). Observe that the first constraint compares the current packet with the last $K$ packets (the ones in the window). For every pair, if $R_p > R_j$, the left constrain in \autoref{eq:quantile} forces $Mg_{pj}$ to be positive and consequently $g_{pj}=1$ (packet $j$ in the window has a lower rank). If $R_p \leq R_j$, the right constraint forces $g_{pj}$ to be 0. \autoref{eq:total_quantile} keeps track of number of packets in the window with rank less than the rank of packet $p$. For packets $p < K$, we add some additional variables that represent the rank of packets arrived and departed before this sequence.

\parab{Deciding to admit or drop.} Recall that \aifo admits the packet if the quantile estimate of the current packet rank is less than the some factor of the available capacity of the queue. We model this as:

{\small
\begin{alignat}{2}
	\hat{c}_p = B\frac{C - \sum_{{\sf pkt}\,\, i < p}{a_i}}{C}& ~~~ \forall{\sf packet}\,p \\
	\hat{c}_p - g_p + \epsilon \leq Za_{p} \leq Z + \hat{c}_p - g_p & ~~~ \forall{\sf packet}\,p
\end{alignat}
}

\noindent where $Z$ is a large constant ($\geq$ maximum of window size and queue size). The first constraint computes the fraction of available capacity multipled by a constant (burst factor in \cite{aifo}) and the second constraint ensure $a_p = 0$ (\ie we drop the packet) if the quantile estimate is higher than the available capacity metric $\hat{c}_p$ and $a_p = 1$ otherwise.

\begin{table}[t!]
	\begin{tabular}{lm{7cm}}
		{\bf Term} & {\bf Meaning} \\\hline
		$C$ & {queue size in the number of packets}\\ \hline
		$a_{p}$ & {$1$ if packet $p$ is admitted, $=0$ if dropped} \\ \hline
		$K$ & number of samples in the window to estimate the quantile\\ \hline
		$g_p$ & {number of packet ranks in the window smaller than rank of packet $p$} \\ \hline
		$g_{pj}$ & {$1$ if rank of packet $j$ in the window is smaller than rank of packet $p$}\\\hline
	\end{tabular}
	\caption{\label{t:aifo_notation} Additional Notations for \aifo.}
\end{table}

Computing the final ordering of packet is similar to \sppifo. These constrains in combination find {\aifo}'s unique solution. All the constraints are also linear. 

\subsection{Proof of \theoremref{th:sp-pifo:avg-gap}}
\label{a:s:sp-pifo-theorem}

\begin{figure}[t]
	\includegraphics{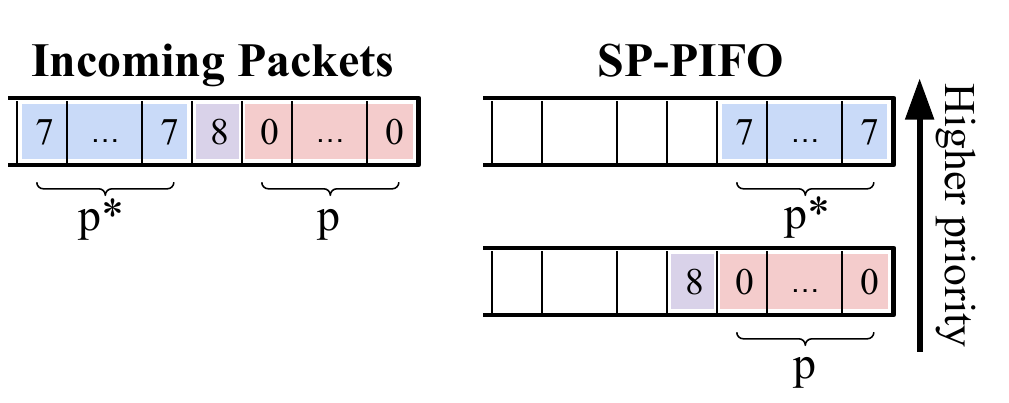}
	\caption{Example of the input trace for \sppifo when $R_{max} = 8$. In this case, \sppifo dequeues all the packets with lowest rank (highest priority) after all the second lowest priority packets ($r=R_{max} - 1$). Packets arrived earlier are on the right size of the queue.\label{f:sp-pifo:theorem}}
\end{figure}

We use the same approach as our proof for \ffdsum. We show a constructive example that matches the gap. In our example, first p packets with the lowest rank (=0) arrive, then 1 packet with the highest possible rank (=$R_{max}$), and finally, $p^{*}$ packets with the second highest possible rank (=$R_{max} - 1$). 

Given this sequence of packets; \sppifo first adds all the $p$ packets with the lowest rank to the lowest priority queue and then, updates the queue rank to 0 (\ie push up). Then, it adds the highest rank packet to the lowest priority queue and updates the queue rank to $R_{max}$. Lowest priority queue can not admit the packets with rank=$R_{max} - 1$ anymore because the condition to admit a packet is that the queue rank should be lower than the packet rank. So, all the $p^*$ packets are enqueued in a higher priority queue. As a result, all these $p^*$ packets are going to be forwarded before the $p$ packets with highest priority (\autoref{f:sp-pifo:theorem} shows this using an example). 

We can compute the weighted sum of packet delays for \sppifo and \pifo is as follows:

{\small
\begin{alignat}{2}
	{\sf Wdelay}_{\pifo} & = \frac{R_{max}p(p-1)}{2} + pp^* + \frac{p^*(p^* - 1)}{2}\\
	{\sf Wdelay}_{\sppifo} & = \frac{p^{*}(p^{*} - 1)}{2} + R_{max}pp^* + \frac{R_{max}p(p - 1)}{2}&
\end{alignat}
}

We can compute the difference in the weighted sum of delays as:

{\small
\begin{alignat}{2}
	{\sf Wdelay}_{\sppifo} - {\sf Wdelay}_{\pifo} & = (R_{max} - 1)pp^{*} \label{eq:28}\\
		& = (R_{max} - 1)(N - 1 - p)p \nonumber
\end{alignat}
}

Note that $p + p^* = N - 1$. We can derive \theoremref{th:sp-pifo:avg-gap} by finding the maximum of \autoref{eq:28}.

\section{List of \metaopt Helper Function}
\label{s:a:helper}
\autoref{t:helper:list} lists the helper functions in \metaopt. \metaopt internally and automatically translates these into constraints. For specific usecases, please refer to \autoref{f:dp-helper-appendix} for \dempin, \autoref{f:ffd-helper-appendix} for \ffd, and \autoref{f:sp-pifo-helper} for \sppifo.

\begin{table*}[h]
	\centering
	\renewcommand{\arraystretch}{1.5}
		\begin{tabular}{lm{11cm}}
			Helper Function & Description \\ \hline
			${\sf IfThen(b, [(x_i, F_i())])}$ & if binary variable $b=1$ then $x_i = F_i()$ for all $i$.
			\\
			${\sf IfThenElse(b, [x_i, F_i()], [(y_j, G_j())])}$ & if binary variable $b=1$ then $x_i = F_i()$ for all $i$, otherwise $y_j = G_j()$ for all j. \\
			${\sf b = AllLeq([x_i], A)}$ & $b = 1$ if all {$x_i$}s are $\leq$ a constant $A$, otherwise $b=0$. \\
			${\sf b = IsLeq(x, y)}$ & $b = 1$ if $x \leq y$, otherwise $b = 0$. \\
			${\sf b = AllEq([x_i], A)}$ & $b=1$ if all ${x_i}$s are $=$ a constant $A$, otherwise $b = 0$. \\
			${\sf b = AND([u_i])}$ & $b=1$ if all ${u_i}s$ are $=1$, otherwise $b = 0$. \\
			${\sf b = OR([u_i])}$ & $b=1$ if at least one $u_i = 1$, otherwise $b = 0$. \\
			${\sf y = Multiplication(u, x)}$ & Linearizes multiplication of a binary variable $b$ and a  continuous variable $x$. \newline (Internally, we choose a simpler encoding if x is non-negative) \\
			${\sf y = MAX([x_i], A)}$ & $y =$ max of ${x_i}$s and a constant $A$. \\
			${\sf y = MIN([x_i], A)}$ & $y =$ min of ${x_j}$s and a constant $A$. \\
			${\sf [b_i] = FindLargestValue([x_i], [u_i])}$ & $b_i = 1$ if $x_i$ is the largest among the group of items $x_j$ with $u_j = 1$, otherwise $b_i = 0$. At least one of one $b_i$ is $=1$.\\
			${\sf [b_i] = FindSmallestValue([x_i], [u_i])}$ & $b_i = 1$ if $x_i$ is the smallest among the group of items $x_j$ with $u_j = 1$, otherwise $b_i = 0$. At least one of one $b_i$ is $=1$.\\
			${\sf r = Rank(y, [x_i])}$ & $r = $ rank of variable $y$ among the group of variables $[x_i]$ (quantile). \\
			${\sf ForceToZeroIfLeq(v, x, y)}$ & Forces $v=0$ if $x \leq y$ (users can model this with ${\sf IfThen}$, but this one is customized and faster). Internally, we choose a simpler encoding if v is binary.
			\\ \hline
		\end{tabular}
	\caption{{\metaopt}'s helper functions.\label{t:helper:list}}
\end{table*}

\section{Black-box search methods}
\label{appendix:black-box}

We next describe our baselines in more detail. We compared \metaopt to these baselines in~\S\ref{s:eval}.

\parab{Random search.} This strawman solution picks random inputs, computes the gap, and returns the input that resulted in the maximum gap after .

\parab{Hill climbing} is a simple local search algorithm. It first randomly chooses an arbitrary input $\mathbf{d}_0$ and then generates its neighbors ($\mathbf{d}_{aux}$): it adds to $\mathbf{d}_0$ a value, which it draws from a zero-mean $\sigma^2$-variance Gaussian distribution. If this neighboring input increases the gap the hill climber moves to it. Otherwise it draws another neighbor. 

The hill climber repeats these steps until it fails to make progress and terminates. This happens when it fails to increase the gap for $K$ consecutive iterations. The hill climber outputs its current solution as a local maximum once it terminates (Algorithm~\ref{alg:hc}).

We re-run the hill climber $M_\mathrm{hc}$ times with different initial inputs and return the solution that produces the maximum gap to minimize the impact of the starting point.


\begin{algorithm}[t]
	\caption{Hill climbing}\label{alg:hc}
	\begin{algorithmic}
		\begin{footnotesize}
			\Input : $\mathbf{d}_0$, $\sigma^2$, $K$
			\EndInput
			\State $\mathbf{d} \gets \mathbf{d}_0$, $\,\,\,k \gets 0$ 
			\While{$k < K$}
			\State $\mathbf{d}_{\mathrm{aux}} \gets \max(\mathbf{d} + \mathbf{z}, \mathbf{0})\,\,\,\,$ where $\,\,\,\,\mathbf{z} \sim \mathcal{N} (\mathbf{0}, \sigma^2 \mathbf{I})$
			\IIf{$\mathrm{gap}(\mathbf{d}_{\mathrm{aux}})>\mathrm{gap}(\mathbf{d})$} $\mathbf{d} \gets \mathbf{d}_{\mathrm{aux}}$, $\,\,\,k \gets -1$
			\EndIIf
			\State $k \gets k+1$			
			\EndWhile
			\Output : $\mathbf{d}$
			\EndOutput
		\end{footnotesize}
	\end{algorithmic}
	\vspace{-1mm}
\end{algorithm}

\parab{Simulated annealing} refines hill-climbing and seeks to avoid getting trapped in a local maxima~\cite{simulated_annealing}. The difference between the two algorithms is simulated annealing may still (with some probability) move to a neighboring input even if that input does not improve the gap.


Simulated annealing gradually decreases the probability of moving to inputs that do not change the gap: it defines a temperature term, $t_p$, which it decreases every $K_p$ iterations to $t_{p+1} = \gamma t_{p}$. Here, $0<\gamma <1$ which ensures $t_p \to 0$. If~$\mathrm{gap}(\mathbf{d}_{\mathrm{aux}}) \leq \mathrm{gap}(\mathbf{d})$, we have $\mathbf{d} \gets \mathbf{d}_{\mathrm{aux}}$ with probability $\exp(\frac{\mathrm{gap}(\mathbf{d}_{\mathrm{aux}}) - \mathrm{gap}(\mathbf{d})}{t_p})$.
%
%
We repeat the process $M_{\mathrm{sa}}$ times and return the best solution. 

\parab{Hill climbing vs simulated annealing.} Hill climbing has less parameters and is better suited for smooth optimizations where there are few local-optima. But simulated annealing is better suited for intricate non-convex optimizations with many local-optima because its exploration phase, although slower, allows it to avoid local optima and works better in the long run.

Both of these algorithms have a number of hyperparameters: we run grid-search to find the ones that produce the highest gap.

\end{document}